\documentclass[10pt]{article}
\usepackage{authblk}
\usepackage[english]{babel}
\usepackage{url}
\usepackage[table, dvipsnames]{xcolor}
\usepackage{colortbl}
\usepackage{graphicx}
\usepackage{amsfonts,amsthm}
\usepackage{amsthm}
\usepackage{amsmath}
\usepackage{amssymb}
\usepackage{float,subcaption}
\usepackage{numprint}
\usepackage{longtable}
\usepackage[showframe=false]{geometry}
\usepackage{changepage}
\usepackage{graphicx}
\usepackage{algorithm}
\usepackage{algpseudocode}
\usepackage{hyperref}

\theoremstyle{plain}
\newtheorem{theorem}{Theorem}

\newtheorem{lemma}[theorem]{Lemma}
\newtheorem{proposition}[theorem]{Proposition}

\theoremstyle{definition}
\newtheorem{definition}[theorem]{Definition}

\theoremstyle{remark}
\newtheorem{remark}[theorem]{Remark}

\newcommand{\qbin}[2]{\begin{bmatrix}{#1}\\ {#2}\end{bmatrix}_q}
\newcommand{\qbincoeff}[3]{\begin{bmatrix}{#1}\\ {#2}\end{bmatrix}_{#3}}

\usepackage{tikz,pgfplots}
\pgfplotsset{compat=1.16}

\tikzset{every node/.style={fill=white}}
\usetikzlibrary{decorations.markings}
\usetikzlibrary{shapes.geometric}
\usetikzlibrary{calc}
\pgfdeclarelayer{edgelayer}
\pgfdeclarelayer{nodelayer}
\pgfsetlayers{edgelayer,nodelayer,main}
\definecolor{lgray}{gray}{.7}
\tikzstyle{none}=[inner sep=0pt]
\tikzstyle{rn}=[circle,fill=Red,draw=black,line width=0.8 pt]
\tikzstyle{gn}=[circle,fill=Lime,draw=black,line width=0.8 pt]
\tikzstyle{yn}=[circle,fill=Yellow,draw=black,line width=0.8 pt]
\tikzstyle{smallCirc}=[circle,fill=White,draw=black]
\tikzstyle{bigCirc}=[circle,fill=White,draw=black]
\tikzstyle{bn}=[inner sep=0pt, circle,fill=black,draw=black,line width=0.8 pt,minimum size=4]
\tikzstyle{simple}=[-,draw=black,line width=1]
\tikzstyle{tick}=[-,
  draw=black,postaction={decorate},
  decoration={markings,mark=at position .5 with {\draw (0,-0.1) -- (0,0.1);}},
  line width=2.000]
\tikzstyle{arrow}=[-,
  draw=black,postaction={decorate},
  decoration={markings,mark=at position .5 with {\arrow{>}}}]
\tikzstyle{arrowHead}=[->,draw=black]

\algnewcommand\algorithmicforeach{\textbf{for each}}
\algdef{S}[FOR]{ForEach}[1]{\algorithmicforeach\ #1\ \algorithmicdo}
\algnewcommand{\IIf}[1]{\State\algorithmicif\ #1\ \algorithmicthen}
\algnewcommand{\EndIIf}{\unskip\ \algorithmicend\ \algorithmicif}
\algnewcommand{\IFor}[2]{\State\algorithmicforeach\ #1\ \algorithmicdo\ #2\ \algorithmicend\ \algorithmicfor}

\begin{document}

\title{Multi-qubit doilies: Enumeration for all ranks and classification for ranks four and five}

\author[1]{Axel Muller}
\author[2]{Metod Saniga}
\author[1]{Alain Giorgetti\footnote{Corresponding author, \texttt{alain.giorgetti@femto-st.fr}}}
\author[3]{Henri de Boutray}
\author[4,5,6]{\\Frédéric Holweck}

\affil[1]{FEMTO-ST institute, Univ. Bourgogne Franche-Comté, CNRS, France}
\affil[2]{Astronomical Institute of the Slovak Academy of Sciences, SK-05960
   Tatransk\'a Lomnica, Slovakia} %
\affil[3]{ColibrITD, France}
\affil[4]{Université de Technologie de Belfort-Montbéliard, F-90010 Belfort cedex, France}
\affil[5]{Laboratoire Interdisciplinaire Carnot de Bourgogne (UMR 6303 - CNRS/ICB/UTBM), France}
\affil[6]{Department of Mathematics and Statistics, Auburn University, Auburn, AL, USA}

\date{}

\maketitle

\begin{abstract}
For $N \geq 2$, an $N$-qubit doily is a doily living in the $N$-qubit symplectic
polar space. These doilies are related to operator-based proofs of
quantum contextuality. Following and extending the strategy of Saniga et al.
(Mathematics 9 (2021) 2272) that focused exclusively on three-qubit doilies, we
first bring forth several formulas giving the number of both linear and
quadratic doilies for any $N > 2$. Then we present an effective algorithm for
the generation of all $N$-qubit doilies. Using this algorithm for $N=4$ and
$N=5$, we provide a classification of $N$-qubit doilies in terms of types of
observables they feature and number of negative lines they are endowed with. We
also list several distinguished findings about $N$-qubit doilies that are absent
in the three-qubit case, point out a couple of specific features exhibited by
linear doilies and outline some prospective extensions of our approach.
\end{abstract}

\section{Introduction}
\label{sec:introduction}

The doily is a remarkable piece of finite geometry that occurs in a number of
disguises. Here, we mention the most prominent ones.

\begin{enumerate}
\item 
\textit{The doily as a duad-syntheme geometry.}
Let us recall a famous Sylvester's construction of the doily~\cite{Syl44}. Given a
six-element set $M_6 \equiv \{1,2,3,4,5,6\}$, a \emph{duad} is an unordered pair
$(ij) \in M_6$, $i \neq j$, and a \emph{syntheme} is a set of three pairwise
disjoint duads, i.\,e. a set $\{(ij),(kl),(mn)\}$ where $i,j,k,l,m,n \in M_6$
are all distinct. The point-line incidence structure whose points are duads and
whose lines are synthemes, with incidence being inclusion, is isomorphic to the
doily, as also illustrated in Figure~\ref{doily-duad-numb}.
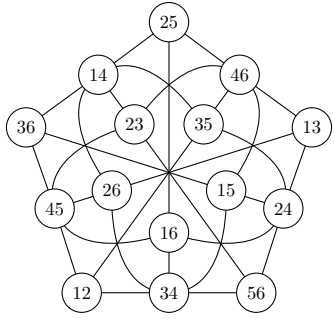
\begin{figure}[ht]
\centerline{
\begin{tikzpicture}[every plot/.style={smooth, tension=2},
  scale=2,
  every node/.style={scale=0.7,fill=white,circle,draw=black}
]
\coordinate (25) at (0,1.0);
\coordinate (34) at (0,-0.80);
\coordinate (16) at (0,-0.40);
\coordinate (36) at (-0.95,0.30);
\coordinate (24) at (0.76,-0.24);
\coordinate (15) at (0.38,-0.12);
\coordinate (12) at (-0.58,-0.80);
\coordinate (46) at (0.47,0.65);
\coordinate (35) at (0.23,0.32);
\coordinate (56) at (0.58,-0.80);
\coordinate (14) at (-0.47,0.65);
\coordinate (23) at (-0.23,0.32);
\coordinate (13) at (0.95,0.30);
\coordinate (45) at (-0.76,-0.24);
\coordinate (26) at (-0.38,-0.12);

\draw (25) -- (14) -- (36);
\draw (36) -- (45) -- (12);
\draw (12) -- (34) -- (56);
\draw (56) -- (24) -- (13);
\draw (13) -- (46) -- (25);
\draw (12) -- (35) -- (46);
\draw (14) -- (23) -- (56);
\draw (13) -- (26) -- (45);
\draw (36) -- (15) -- (24);
\draw (25) -- (16) -- (34);
\draw plot coordinates{(26) (14) (35)};
\draw plot coordinates{(23) (46) (15)};
\draw plot coordinates{(35) (24) (16)};
\draw plot coordinates{(15) (34) (26)};
\draw plot coordinates{(16) (45) (23)};

\node at (25) {$25$};
\node at (34) {$34$};
\node at (16) {$16$};
\node at (36) {$36$};
\node at (24) {$24$};
\node at (15) {$15$};
\node at (12) {$12$};
\node at (46) {$46$};
\node at (35) {$35$};
\node at (56) {$56$};
\node at (14) {$14$};
\node at (23) {$23$};
\node at (13) {$13$};
\node at (45) {$45$};
\node at (26) {$26$};
\end{tikzpicture} }
\vspace*{.2cm}
\caption{A duad-syntheme model of the doily.}
\label{doily-duad-numb}
\end{figure}

\item
\textit{The doily as the Cremona--Richmond configuration.}
It is a particular $15_3$-configuration, i.\,e. a self-dual configuration of 15
points and 15 lines, with three points on  a line and, dually, three lines
through a point such that it contains no triangles~\cite{Cre77,Ric00}. Up to
isomorphism, there are altogether 245,342 $15_3$-configurations, of which only
the doily enjoys the property of being triangle-free.

\item 
\textit{The doily as a generalized quadrangle.}
A generalized quadrangle GQ$(s,t)$ of order $(s,t)$ is an incidence structure of
points and lines (blocks) where every point is on $t+1$ lines ($t>0$), and every
line contains $s+1$ points ($s>0$) such that if $p$ is a point and $L$ is a
line, $p$ not on $L$, then there is a unique point $q$ on $L$ such that $p$ and
$q$ are collinear. The doily is isomorphic to the unique generalized quadrangle
with $s=t=2$~\cite{PT09}.

\item
\textit{The doily as a symplectic polar space.}
Given a $d$-dimensional projective space PG$(d,2)$ over the two-elements field 
$\mathbb{F}_2 = \{0,1\}$ of modulo-2 arithmetic, a \emph{polar
space} ${\cal P}$ in this projective space consists of the projective subspaces
that are \emph{totally isotropic/singular} with respect to a given non-singular
bilinear form~\cite{HT91,Cam92}; PG$(d,2)$ is called the \emph{ambient projective
space} of ${\cal P}$.
A projective subspace of maximal dimension in ${\cal P}$ is called a {\it
generator}; all generators have the same (projective) dimension $r - 1$.
One calls $r$ the \emph{rank} of the polar space.
The \emph{symplectic polar space} $\mathcal{W}(2N-1,2)$, $N \geq 1$, consists of
all the points of PG$(2N-1, 2)$, $\{(x_1, x_2, \ldots, x_{2N}): x_j \in \{0,1\},
j \in \{1,2, \ldots, 2N \}\}\backslash \{(0, 0, \ldots, 0)\}$,
 together with the totally isotropic subspaces with respect to the standard
 symplectic form
\begin{equation} 
\sigma(x,y) = x_1y_{N+1} - x_{N+1}y_1 + x_2y_{N+2} - x_{N+2}y_2 + \dots 
              + x_{N} y_{2N} - x_{2N} y_{N}.
\label{symplf}
\end{equation}
Throughout the paper, the space name $\mathcal{W}(2N-1,2)$ is often shortened as $W_N$.
This space features
\begin{equation*}
|W_N|_p = 4^N - 1
\label{ptsinwn}
\end{equation*}
points and
\begin{equation*}
|W_N|_g = (2+1)(2^2+1)\cdots(2^N+1)
\label{geninwn}
\end{equation*}
generators. The doily is isomorphic to the symplectic polar space of rank $N=2$, 
$\mathcal{W}(3,2)$.

\item
\textit{Multi-qubit doilies.}
This paper is about doilies related to Kochen--Specker operator-based proofs of
quantum contextuality, to be called $N$-qubit doilies or multi-qubit doilies.
We follow the terminology and notation of Section 2
of~\cite{SdHG21}, to which the reader can refer for more finite-geometric
background. Let
\begin{equation*}
X = \left(
\begin{array}{rr}
0 & 1 \\
1 & 0 \\
\end{array}
\right),~~
Y = \left(
\begin{array}{rr}
0 & -i \\
i & 0 \\
\end{array}
\right)~~{\rm and}~~
Z = \left(
\begin{array}{rr}
1 & 0 \\
0 & -1 \\
\end{array}
\right)
\label{paulis}
\end{equation*}
be the Pauli matrices, $I$ the identity matrix, `$\otimes$' the tensor product of
matrices and  ${\cal I}_N\equiv I_{(1)}\otimes I_{(2)}\otimes\ldots\otimes 
I_{(N)}$, and let ${\cal S}_N = \{G_1 \otimes G_2 \otimes \cdots \otimes G_N:~ 
G_j \in \{I, X, Y, Z \},~ j \in \{1, 2, \ldots, N \}\} \backslash\{{\cal I}_N\}$. 
The $4^N-1$ $N$-qubit observables of ${\cal S}_N$ can be bijectively identified
with the $4^N-1$ points of $\mathcal{W}(2N-1,2)$ in such a way that any two
commuting observables  are represented by collinear points and the product of
the three observables lying on a line of $\mathcal{W}(2N-1,2)$ is $+{\cal I}_N$
or $-{\cal I}_N$ (see, for example,~\cite[Section 5.3.2]{Hol19}). If the
symplectic form in the ambient space PG$(2N-1, 2)$, defining 
$\mathcal{W}(2N-1,2)$, is  given by Eq.~(\ref{symplf}), then the corresponding
bijection reads
\begin{equation}
G_j \leftrightarrow (x_j, x_{j+N}),~j \in \{1, 2,\dots,N\},
\label{nobspts}
\end{equation}
with the assumption that
\begin{equation}
I\leftrightarrow(0,0),~X\leftrightarrow (0,1),~Y\leftrightarrow(1,1),~{\rm and}~
Z\leftrightarrow (1,0).
\label{paulipts}
\end{equation}
To briefly illustrate this property, let us consider the three-qubit 
$\mathcal{W}(5,2)$ and one of its lines, say $(0,1,1;1,1,0)$, $(1,0,0;0,0,1)$ and 
$(1,1,1;1,1,1)$.
Using the correspondences~(\ref{nobspts}) and~(\ref{paulipts}) we find that the corresponding 
observables are $X \otimes Y \otimes Z$, $Z \otimes I \otimes X$ and  $Y \otimes
Y \otimes Y$, respectively; these observables indeed pairwise commute and their 
product is $+ I \otimes I \otimes I$. 

In what follows, $W_N$ will always be understood as having its
points labeled by the  $N$-qubit observables as described above, and any doily
lying in it, together with the inherited labeling, will be called an $N$-qubit
doily ($N \geq 2$). Slightly rephrased, an $N$\emph{-qubit doily} is a doily
whose points are bijectively identified with 15 specific observables from ${\cal
S}_N$, such that any two commuting observables share the same line, and, given any
line, the product of (any) two observables lying on it is, up to a sign, equal
to the remaining observable on it. A line of an $N$-qubit doily will be called
\emph{positive} (resp. \emph{negative}) if the product of its three observables
is $+{\cal I}_N$ (resp. $-{\cal I}_N$). To avoid any possible misunderstanding,
it is worth mentioning that the product of observables is the (ordinary) matrix
product, denoted by a dot ($.$), induced by the following multiplication table
of Pauli matrices.

\begin{center}
  \begin{tabular}{c|ccc}
    $.$  %
       & $X$  & $Y$  & $Z$  \\
    \hline
    $X$ %
       & $I$       & $iZ$ & $-iY$\\
    $Y$%
        & $-iZ$& $I$       & $iX$ \\
    $Z$%
      & $iY$ & $-iX$& $I$       \\
  \end{tabular}
\end{center}

From here on, the geometrical points are considered to be finite words on the
four-letter alphabet $\{I,X,Y,Z\}$ that encode the observables $G_1 \otimes G_2
\otimes \cdots \otimes G_N$, while omitting the symbol $\otimes$ for the tensor
product and forgetting in the sequel about the matrix nature of $I$, $X$, $Y$
and $Z$.
\end{enumerate}

\paragraph{Contributions and paper outline.} 
Our contributions start in Section~\ref{basic-sec}, with a presentation of
several facts about $N$-qubit doilies that motivates the design of an effective
algorithm to enumerate $N$-qubit doilies for any rank $N$
(Section~\ref{sec:enum-sec}). By geometric considerations, we first establish in
Section~\ref{nb-sec} closed formulas for the numbers of $N$-qubit doilies. As
these numbers increase rapidly with $N$, the enumeration algorithm can in
practice only be executed for small numbers of qubits. We use it in
Section~\ref{classify-sec} to classify $N$-qubit doilies for $N=4$ and $N=5$,
according to their types of observables and their configurations of negative
lines. We thus produce precise tables for the number of doilies in each
category/class, reproduced in the appendices of this paper.
Section~\ref{classify-sec} also analyzes these results and points out various
findings about $N$-qubit doilies that are absent in the known three-qubit case.
Section~\ref{conclusion-sec} concludes and outlines some prospective extensions
of our approach.

\section{Some basic facts about multi-qubit doilies}
\label{basic-sec}

\subsection{Patterns formed by negative lines}

It is a straightforward task to work out possible types of patterns of negative
lines an $N$-qubit doily can be endowed with. This classification follows
readily from the facts that each grid in the doily  must contain an odd number
of negative lines and that two different grids have two intersecting lines in
common. And as a grid has an even number of lines the types of configurations
come in complementary pairs, as depicted in Figure~\ref{doily-nl}.

\begin{figure}[!ht]
\centering
\begin{minipage}{0.4\textwidth}
\begin{subfigure}{0.4\textwidth}
\begin{tikzpicture}[every plot/.style={smooth, tension=2},
  scale=1,
  every node/.style={scale=0.7,fill=white,circle,draw=black}
]
\node[scale=1.3,draw=white] at (-.75,1) {3};

\coordinate (25) at (0,1.0);
\coordinate (34) at (0,-0.80);
\coordinate (16) at (0,-0.40);
\coordinate (36) at (-0.95,0.30);
\coordinate (24) at (0.76,-0.24);
\coordinate (15) at (0.38,-0.12);
\coordinate (12) at (-0.58,-0.80);
\coordinate (46) at (0.47,0.65);
\coordinate (35) at (0.23,0.32);
\coordinate (56) at (0.58,-0.80);
\coordinate (14) at (-0.47,0.65);
\coordinate (23) at (-0.23,0.32);
\coordinate (13) at (0.95,0.30);
\coordinate (45) at (-0.76,-0.24);
\coordinate (26) at (-0.38,-0.12);

\draw (25) -- (14) -- (36);
\draw[very thick] (36) -- (45) -- (12);
\draw (12) -- (34) -- (56);
\draw[very thick] (56) -- (24) -- (13);
\draw (13) -- (46) -- (25);
\draw (12) -- (35) -- (46);
\draw (14) -- (23) -- (56);
\draw (13) -- (26) -- (45);
\draw (36) -- (15) -- (24);
\draw (25) -- (16) -- (34);
\draw plot coordinates{(26) (14) (35)};
\draw plot coordinates{(23) (46) (15)};
\draw plot coordinates{(35) (24) (16)};
\draw[very thick] plot coordinates{(15) (34) (26)};
\draw plot coordinates{(16) (45) (23)};

\node at (25) {};
\node at (34) {};
\node at (16) {};
\node at (36) {};
\node at (24) {};
\node at (15) {};
\node at (12) {};
\node at (46) {};
\node at (35) {};
\node at (56) {};
\node at (14) {};
\node at (23) {};
\node at (13) {};
\node at (45) {};
\node at (26) {};
\end{tikzpicture}
\end{subfigure}
\begin{subfigure}{0.4\textwidth}
\begin{tikzpicture}[every plot/.style={smooth, tension=2},
  scale=1,
  every node/.style={scale=0.7,fill=white,circle,draw=black}
]
\node[scale=1.3,draw=white] at (-.75,1) {12};

\coordinate (25) at (0,1.0);
\coordinate (34) at (0,-0.80);
\coordinate (16) at (0,-0.40);
\coordinate (36) at (-0.95,0.30);
\coordinate (24) at (0.76,-0.24);
\coordinate (15) at (0.38,-0.12);
\coordinate (12) at (-0.58,-0.80);
\coordinate (46) at (0.47,0.65);
\coordinate (35) at (0.23,0.32);
\coordinate (56) at (0.58,-0.80);
\coordinate (14) at (-0.47,0.65);
\coordinate (23) at (-0.23,0.32);
\coordinate (13) at (0.95,0.30);
\coordinate (45) at (-0.76,-0.24);
\coordinate (26) at (-0.38,-0.12);

\draw[very thick] (25) -- (14) -- (36);
\draw (36) -- (45) -- (12);
\draw[very thick] (12) -- (34) -- (56);
\draw (56) -- (24) -- (13);
\draw[very thick] (13) -- (46) -- (25);
\draw[very thick] (12) -- (35) -- (46);
\draw[very thick] (14) -- (23) -- (56);
\draw[very thick] (13) -- (26) -- (45);
\draw[very thick] (36) -- (15) -- (24);
\draw[very thick] (25) -- (16) -- (34);
\draw[very thick] plot coordinates{(26) (14) (35)};
\draw[very thick] plot coordinates{(23) (46) (15)};
\draw[very thick] plot coordinates{(35) (24) (16)};
\draw plot coordinates{(15) (34) (26)};
\draw[very thick] plot coordinates{(16) (45) (23)};

\node[very thick] at (25) {};
\node at (34) {};
\node[very thick] at (16) {};
\node at (36) {};
\node at (24) {};
\node at (15) {};
\node at (12) {};
\node[very thick] at (46) {};
\node[very thick] at (35) {};
\node at (56) {};
\node[very thick] at (14) {};
\node[very thick] at (23) {};
\node at (13) {};
\node at (45) {};
\node at (26) {};
\end{tikzpicture}
\end{subfigure}\\
\begin{subfigure}{0.4\textwidth}
\begin{tikzpicture}[every plot/.style={smooth, tension=2},
  scale=1,
  every node/.style={scale=0.7,fill=white,circle,draw=black}
]
\node[scale=1.3,draw=white] at (-.75,1) {4};

\coordinate (25) at (0,1.0);
\coordinate (34) at (0,-0.80);
\coordinate (16) at (0,-0.40);
\coordinate (36) at (-0.95,0.30);
\coordinate (24) at (0.76,-0.24);
\coordinate (15) at (0.38,-0.12);
\coordinate (12) at (-0.58,-0.80);
\coordinate (46) at (0.47,0.65);
\coordinate (35) at (0.23,0.32);
\coordinate (56) at (0.58,-0.80);
\coordinate (14) at (-0.47,0.65);
\coordinate (23) at (-0.23,0.32);
\coordinate (13) at (0.95,0.30);
\coordinate (45) at (-0.76,-0.24);
\coordinate (26) at (-0.38,-0.12);

\draw (25) -- (14) -- (36);
\draw[very thick] (36) -- (45) -- (12);
\draw[very thick] (12) -- (34) -- (56);
\draw[very thick] (56) -- (24) -- (13);
\draw (13) -- (46) -- (25);
\draw (12) -- (35) -- (46);
\draw (14) -- (23) -- (56);
\draw (13) -- (26) -- (45);
\draw (36) -- (15) -- (24);
\draw[very thick] (25) -- (16) -- (34);
\draw plot coordinates{(26) (14) (35)};
\draw plot coordinates{(23) (46) (15)};
\draw plot coordinates{(35) (24) (16)};
\draw plot coordinates{(15) (34) (26)};
\draw plot coordinates{(16) (45) (23)};

\node at (25) {};
\node at (34) {};
\node at (16) {};
\node at (36) {};
\node at (24) {};
\node at (15) {};
\node at (12) {};
\node at (46) {};
\node at (35) {};
\node at (56) {};
\node at (14) {};
\node at (23) {};
\node at (13) {};
\node at (45) {};
\node at (26) {};
\end{tikzpicture}
\end{subfigure}
\begin{subfigure}{0.4\textwidth}
\begin{tikzpicture}[every plot/.style={smooth, tension=2},
  scale=1,
  every node/.style={scale=0.7,fill=white,circle,draw=black}
]
\node[scale=1.3,draw=white] at (-.75,1) {11};

\coordinate (25) at (0,1.0);
\coordinate (34) at (0,-0.80);
\coordinate (16) at (0,-0.40);
\coordinate (36) at (-0.95,0.30);
\coordinate (24) at (0.76,-0.24);
\coordinate (15) at (0.38,-0.12);
\coordinate (12) at (-0.58,-0.80);
\coordinate (46) at (0.47,0.65);
\coordinate (35) at (0.23,0.32);
\coordinate (56) at (0.58,-0.80);
\coordinate (14) at (-0.47,0.65);
\coordinate (23) at (-0.23,0.32);
\coordinate (13) at (0.95,0.30);
\coordinate (45) at (-0.76,-0.24);
\coordinate (26) at (-0.38,-0.12);

\draw[very thick] (25) -- (14) -- (36);
\draw (36) -- (45) -- (12);
\draw (12) -- (34) -- (56);
\draw (56) -- (24) -- (13);
\draw[very thick] (13) -- (46) -- (25);
\draw[very thick] (12) -- (35) -- (46);
\draw[very thick] (14) -- (23) -- (56);
\draw[very thick] (13) -- (26) -- (45);
\draw[very thick] (36) -- (15) -- (24);
\draw (25) -- (16) -- (34);
\draw[very thick] plot coordinates{(26) (14) (35)};
\draw[very thick] plot coordinates{(23) (46) (15)};
\draw[very thick] plot coordinates{(35) (24) (16)};
\draw[very thick] plot coordinates{(15) (34) (26)};
\draw[very thick] plot coordinates{(16) (45) (23)};

\node at (25) {};
\node at (34) {};
\node at (16) {};
\node at (36) {};
\node at (24) {};
\node[very thick] at (15) {};
\node at (12) {};
\node[very thick] at (46) {};
\node[very thick] at (35) {};
\node at (56) {};
\node[very thick] at (14) {};
\node[very thick] at (23) {};
\node at (13) {};
\node at (45) {};
\node[very thick] at (26) {};
\end{tikzpicture}
\end{subfigure}\\
\begin{subfigure}{0.4\textwidth}
\begin{tikzpicture}[every plot/.style={smooth, tension=2},
  scale=1,
  every node/.style={scale=0.7,fill=white,circle,draw=black}
]
\node[scale=1.3,draw=white] at (-.75,1) {5};

\coordinate (25) at (0,1.0);
\coordinate (34) at (0,-0.80);
\coordinate (16) at (0,-0.40);
\coordinate (36) at (-0.95,0.30);
\coordinate (24) at (0.76,-0.24);
\coordinate (15) at (0.38,-0.12);
\coordinate (12) at (-0.58,-0.80);
\coordinate (46) at (0.47,0.65);
\coordinate (35) at (0.23,0.32);
\coordinate (56) at (0.58,-0.80);
\coordinate (14) at (-0.47,0.65);
\coordinate (23) at (-0.23,0.32);
\coordinate (13) at (0.95,0.30);
\coordinate (45) at (-0.76,-0.24);
\coordinate (26) at (-0.38,-0.12);

\draw[very thick] (25) -- (14) -- (36);
\draw[very thick] (36) -- (45) -- (12);
\draw[very thick] (12) -- (34) -- (56);
\draw[very thick] (56) -- (24) -- (13);
\draw[very thick] (13) -- (46) -- (25);
\draw (12) -- (35) -- (46);
\draw (14) -- (23) -- (56);
\draw (13) -- (26) -- (45);
\draw (36) -- (15) -- (24);
\draw (25) -- (16) -- (34);
\draw plot coordinates{(26) (14) (35)};
\draw plot coordinates{(23) (46) (15)};
\draw plot coordinates{(35) (24) (16)};
\draw plot coordinates{(15) (34) (26)};
\draw plot coordinates{(16) (45) (23)};

\node at (25) {};
\node at (34) {};
\node at (16) {};
\node at (36) {};
\node at (24) {};
\node at (15) {};
\node at (12) {};
\node at (46) {};
\node at (35) {};
\node at (56) {};
\node at (14) {};
\node at (23) {};
\node at (13) {};
\node at (45) {};
\node at (26) {};
\end{tikzpicture}
\end{subfigure}
\begin{subfigure}{0.4\textwidth}
\begin{tikzpicture}[every plot/.style={smooth, tension=2},
  scale=1,
  every node/.style={scale=0.7,fill=white,circle,draw=black}
]
\node[scale=1.3,draw=white] at (-.75,1) {10};

\coordinate (25) at (0,1.0);
\coordinate (34) at (0,-0.80);
\coordinate (16) at (0,-0.40);
\coordinate (36) at (-0.95,0.30);
\coordinate (24) at (0.76,-0.24);
\coordinate (15) at (0.38,-0.12);
\coordinate (12) at (-0.58,-0.80);
\coordinate (46) at (0.47,0.65);
\coordinate (35) at (0.23,0.32);
\coordinate (56) at (0.58,-0.80);
\coordinate (14) at (-0.47,0.65);
\coordinate (23) at (-0.23,0.32);
\coordinate (13) at (0.95,0.30);
\coordinate (45) at (-0.76,-0.24);
\coordinate (26) at (-0.38,-0.12);

\draw (25) -- (14) -- (36);
\draw (36) -- (45) -- (12);
\draw (12) -- (34) -- (56);
\draw (56) -- (24) -- (13);
\draw (13) -- (46) -- (25);
\draw[very thick] (12) -- (35) -- (46);
\draw[very thick] (14) -- (23) -- (56);
\draw[very thick] (13) -- (26) -- (45);
\draw[very thick] (36) -- (15) -- (24);
\draw[very thick] (25) -- (16) -- (34);
\draw[very thick] plot coordinates{(26) (14) (35)};
\draw[very thick] plot coordinates{(23) (46) (15)};
\draw[very thick] plot coordinates{(35) (24) (16)};
\draw[very thick] plot coordinates{(15) (34) (26)};
\draw[very thick] plot coordinates{(16) (45) (23)};

\node at (25) {};
\node at (34) {};
\node[very thick] at (16) {};
\node at (36) {};
\node at (24) {};
\node[very thick] at (15) {};
\node at (12) {};
\node at (46) {};
\node[very thick] at (35) {};
\node at (56) {};
\node at (14) {};
\node[very thick] at (23) {};
\node at (13) {};
\node at (45) {};
\node[very thick] at (26) {};
\end{tikzpicture}
\end{subfigure}
\end{minipage} \hspace*{1cm} 
\begin{minipage}{0.4\textwidth}
\begin{subfigure}{0.4\textwidth}
\begin{tikzpicture}[every plot/.style={smooth, tension=2},
  scale=1,
  every node/.style={scale=0.7,fill=white,circle,draw=black}
]
\node[scale=1.3,draw=white] at (-.75,1) {6};

\coordinate (25) at (0,1.0);
\coordinate (34) at (0,-0.80);
\coordinate (16) at (0,-0.40);
\coordinate (36) at (-0.95,0.30);
\coordinate (24) at (0.76,-0.24);
\coordinate (15) at (0.38,-0.12);
\coordinate (12) at (-0.58,-0.80);
\coordinate (46) at (0.47,0.65);
\coordinate (35) at (0.23,0.32);
\coordinate (56) at (0.58,-0.80);
\coordinate (14) at (-0.47,0.65);
\coordinate (23) at (-0.23,0.32);
\coordinate (13) at (0.95,0.30);
\coordinate (45) at (-0.76,-0.24);
\coordinate (26) at (-0.38,-0.12);

\draw[very thick] (25) -- (14) -- (36);
\draw[very thick] (36) -- (45) -- (12);
\draw (12) -- (34) -- (56);
\draw[very thick] (56) -- (24) -- (13);
\draw[very thick] (13) -- (46) -- (25);
\draw (12) -- (35) -- (46);
\draw (14) -- (23) -- (56);
\draw (13) -- (26) -- (45);
\draw (36) -- (15) -- (24);
\draw[very thick] (25) -- (16) -- (34);
\draw plot coordinates{(26) (14) (35)};
\draw plot coordinates{(23) (46) (15)};
\draw plot coordinates{(35) (24) (16)};
\draw[very thick] plot coordinates{(15) (34) (26)};
\draw plot coordinates{(16) (45) (23)};

\node[very thick] at (25) {};
\node at (34) {};
\node at (16) {};
\node at (36) {};
\node at (24) {};
\node at (15) {};
\node at (12) {};
\node at (46) {};
\node at (35) {};
\node at (56) {};
\node at (14) {};
\node at (23) {};
\node at (13) {};
\node at (45) {};
\node at (26) {};
\end{tikzpicture}
\end{subfigure}
\begin{subfigure}{0.4\textwidth}
\begin{tikzpicture}[every plot/.style={smooth, tension=2},
  scale=1,
  every node/.style={scale=0.7,fill=white,circle,draw=black}
]
\node[scale=1.3,draw=white] at (-.75,1) {9};

\coordinate (25) at (0,1.0);
\coordinate (34) at (0,-0.80);
\coordinate (16) at (0,-0.40);
\coordinate (36) at (-0.95,0.30);
\coordinate (24) at (0.76,-0.24);
\coordinate (15) at (0.38,-0.12);
\coordinate (12) at (-0.58,-0.80);
\coordinate (46) at (0.47,0.65);
\coordinate (35) at (0.23,0.32);
\coordinate (56) at (0.58,-0.80);
\coordinate (14) at (-0.47,0.65);
\coordinate (23) at (-0.23,0.32);
\coordinate (13) at (0.95,0.30);
\coordinate (45) at (-0.76,-0.24);
\coordinate (26) at (-0.38,-0.12);

\draw (25) -- (14) -- (36);
\draw (36) -- (45) -- (12);
\draw[very thick] (12) -- (34) -- (56);
\draw (56) -- (24) -- (13);
\draw (13) -- (46) -- (25);
\draw[very thick] (12) -- (35) -- (46);
\draw[very thick] (14) -- (23) -- (56);
\draw[very thick] (13) -- (26) -- (45);
\draw[very thick] (36) -- (15) -- (24);
\draw (25) -- (16) -- (34);
\draw[very thick] plot coordinates{(26) (14) (35)};
\draw[very thick] plot coordinates{(23) (46) (15)};
\draw[very thick] plot coordinates{(35) (24) (16)};
\draw plot coordinates{(15) (34) (26)};
\draw[very thick] plot coordinates{(16) (45) (23)};

\node at (25) {};
\node at (34) {};
\node at (16) {};
\node at (36) {};
\node at (24) {};
\node at (15) {};
\node at (12) {};
\node at (46) {};
\node[very thick] at (35) {};
\node at (56) {};
\node at (14) {};
\node[very thick] at (23) {};
\node at (13) {};
\node at (45) {};
\node at (26) {};
\end{tikzpicture}
\end{subfigure}\\
\begin{subfigure}{0.4\textwidth}
\begin{tikzpicture}[every plot/.style={smooth, tension=2},
  scale=1,
  every node/.style={scale=0.7,fill=white,circle,draw=black}
]
\node[scale=1.3,draw=white] at (-.75,1) {7A};

\coordinate (25) at (0,1.0);
\coordinate (34) at (0,-0.80);
\coordinate (16) at (0,-0.40);
\coordinate (36) at (-0.95,0.30);
\coordinate (24) at (0.76,-0.24);
\coordinate (15) at (0.38,-0.12);
\coordinate (12) at (-0.58,-0.80);
\coordinate (46) at (0.47,0.65);
\coordinate (35) at (0.23,0.32);
\coordinate (56) at (0.58,-0.80);
\coordinate (14) at (-0.47,0.65);
\coordinate (23) at (-0.23,0.32);
\coordinate (13) at (0.95,0.30);
\coordinate (45) at (-0.76,-0.24);
\coordinate (26) at (-0.38,-0.12);

\draw[very thick] (25) -- (14) -- (36);
\draw (36) -- (45) -- (12);
\draw[very thick] (12) -- (34) -- (56);
\draw (56) -- (24) -- (13);
\draw[very thick] (13) -- (46) -- (25);
\draw (12) -- (35) -- (46);
\draw (14) -- (23) -- (56);
\draw (13) -- (26) -- (45);
\draw (36) -- (15) -- (24);
\draw (25) -- (16) -- (34);
\draw[very thick] plot coordinates{(26) (14) (35)};
\draw[very thick] plot coordinates{(23) (46) (15)};
\draw[very thick] plot coordinates{(35) (24) (16)};
\draw plot coordinates{(15) (34) (26)};
\draw[very thick] plot coordinates{(16) (45) (23)};

\node at (25) {};
\node at (34) {};
\node at (16) {};
\node at (36) {};
\node at (24) {};
\node at (15) {};
\node at (12) {};
\node at (46) {};
\node at (35) {};
\node at (56) {};
\node at (14) {};
\node at (23) {};
\node at (13) {};
\node at (45) {};
\node at (26) {};
\end{tikzpicture}
\end{subfigure}
\begin{subfigure}{0.4\textwidth}
\begin{tikzpicture}[every plot/.style={smooth, tension=2},
  scale=1,
  every node/.style={scale=0.7,fill=white,circle,draw=black}
]
\node[scale=1.3,draw=white] at (-.75,1) {8A};

\coordinate (25) at (0,1.0);
\coordinate (34) at (0,-0.80);
\coordinate (16) at (0,-0.40);
\coordinate (36) at (-0.95,0.30);
\coordinate (24) at (0.76,-0.24);
\coordinate (15) at (0.38,-0.12);
\coordinate (12) at (-0.58,-0.80);
\coordinate (46) at (0.47,0.65);
\coordinate (35) at (0.23,0.32);
\coordinate (56) at (0.58,-0.80);
\coordinate (14) at (-0.47,0.65);
\coordinate (23) at (-0.23,0.32);
\coordinate (13) at (0.95,0.30);
\coordinate (45) at (-0.76,-0.24);
\coordinate (26) at (-0.38,-0.12);

\draw (25) -- (14) -- (36);
\draw[very thick] (36) -- (45) -- (12);
\draw (12) -- (34) -- (56);
\draw[very thick] (56) -- (24) -- (13);
\draw (13) -- (46) -- (25);
\draw[very thick] (12) -- (35) -- (46);
\draw[very thick] (14) -- (23) -- (56);
\draw[very thick] (13) -- (26) -- (45);
\draw[very thick] (36) -- (15) -- (24);
\draw[very thick] (25) -- (16) -- (34);
\draw plot coordinates{(26) (14) (35)};
\draw plot coordinates{(23) (46) (15)};
\draw plot coordinates{(35) (24) (16)};
\draw[very thick] plot coordinates{(15) (34) (26)};
\draw plot coordinates{(16) (45) (23)};

\node at (25) {};
\node at (34) {};
\node at (16) {};
\node at (36) {};
\node at (24) {};
\node at (15) {};
\node at (12) {};
\node at (46) {};
\node at (35) {};
\node at (56) {};
\node at (14) {};
\node at (23) {};
\node at (13) {};
\node at (45) {};
\node at (26) {};
\end{tikzpicture}
\end{subfigure}\\
\begin{subfigure}{0.4\textwidth}
\begin{tikzpicture}[every plot/.style={smooth, tension=2},
  scale=1,
  every node/.style={scale=0.7,fill=white,circle,draw=black}
]
\node[scale=1.3,draw=white] at (-.75,1) {7B};

\coordinate (25) at (0,1.0);
\coordinate (34) at (0,-0.80);
\coordinate (16) at (0,-0.40);
\coordinate (36) at (-0.95,0.30);
\coordinate (24) at (0.76,-0.24);
\coordinate (15) at (0.38,-0.12);
\coordinate (12) at (-0.58,-0.80);
\coordinate (46) at (0.47,0.65);
\coordinate (35) at (0.23,0.32);
\coordinate (56) at (0.58,-0.80);
\coordinate (14) at (-0.47,0.65);
\coordinate (23) at (-0.23,0.32);
\coordinate (13) at (0.95,0.30);
\coordinate (45) at (-0.76,-0.24);
\coordinate (26) at (-0.38,-0.12);

\draw (25) -- (14) -- (36);
\draw[very thick] (36) -- (45) -- (12);
\draw[very thick] (12) -- (34) -- (56);
\draw[very thick] (56) -- (24) -- (13);
\draw (13) -- (46) -- (25);
\draw[very thick] (12) -- (35) -- (46);
\draw[very thick] (14) -- (23) -- (56);
\draw[very thick] (13) -- (26) -- (45);
\draw[very thick] (36) -- (15) -- (24);
\draw (25) -- (16) -- (34);
\draw plot coordinates{(26) (14) (35)};
\draw plot coordinates{(23) (46) (15)};
\draw plot coordinates{(35) (24) (16)};
\draw plot coordinates{(15) (34) (26)};
\draw plot coordinates{(16) (45) (23)};

\node at (25) {};
\node at (34) {};
\node at (16) {};
\node at (36) {};
\node at (24) {};
\node at (15) {};
\node[very thick] at (12) {};
\node at (46) {};
\node at (35) {};
\node[very thick] at (56) {};
\node at (14) {};
\node at (23) {};
\node at (13) {};
\node at (45) {};
\node at (26) {};
\end{tikzpicture}
\end{subfigure}
\begin{subfigure}{0.4\textwidth}
\begin{tikzpicture}[every plot/.style={smooth, tension=2},
  scale=1,
  every node/.style={scale=0.7,fill=white,circle,draw=black}
]
\node[scale=1.3,draw=white] at (-.75,1) {8B};

\coordinate (25) at (0,1.0);
\coordinate (34) at (0,-0.80);
\coordinate (16) at (0,-0.40);
\coordinate (36) at (-0.95,0.30);
\coordinate (24) at (0.76,-0.24);
\coordinate (15) at (0.38,-0.12);
\coordinate (12) at (-0.58,-0.80);
\coordinate (46) at (0.47,0.65);
\coordinate (35) at (0.23,0.32);
\coordinate (56) at (0.58,-0.80);
\coordinate (14) at (-0.47,0.65);
\coordinate (23) at (-0.23,0.32);
\coordinate (13) at (0.95,0.30);
\coordinate (45) at (-0.76,-0.24);
\coordinate (26) at (-0.38,-0.12);

\draw[very thick] (25) -- (14) -- (36);
\draw (36) -- (45) -- (12);
\draw (12) -- (34) -- (56);
\draw (56) -- (24) -- (13);
\draw[very thick] (13) -- (46) -- (25);
\draw (12) -- (35) -- (46);
\draw (14) -- (23) -- (56);
\draw (13) -- (26) -- (45);
\draw (36) -- (15) -- (24);
\draw[very thick] (25) -- (16) -- (34);
\draw[very thick] plot coordinates{(26) (14) (35)};
\draw[very thick] plot coordinates{(23) (46) (15)};
\draw[very thick] plot coordinates{(35) (24) (16)};
\draw[very thick] plot coordinates{(15) (34) (26)};
\draw[very thick] plot coordinates{(16) (45) (23)};

\node[very thick] at (25) {};
\node at (34) {};
\node[very thick] at (16) {};
\node at (36) {};
\node at (24) {};
\node at (15) {};
\node at (12) {};
\node at (46) {};
\node at (35) {};
\node at (56) {};
\node at (14) {};
\node at (23) {};
\node at (13) {};
\node at (45) {};
\node at (26) {};
\end{tikzpicture}
\end{subfigure}
\end{minipage} \caption{Generic representatives of the twelve different types of configurations of 
  negative lines (bold) that can be found in a multi-qubit doily.}
\label{doily-nl}
\end{figure}
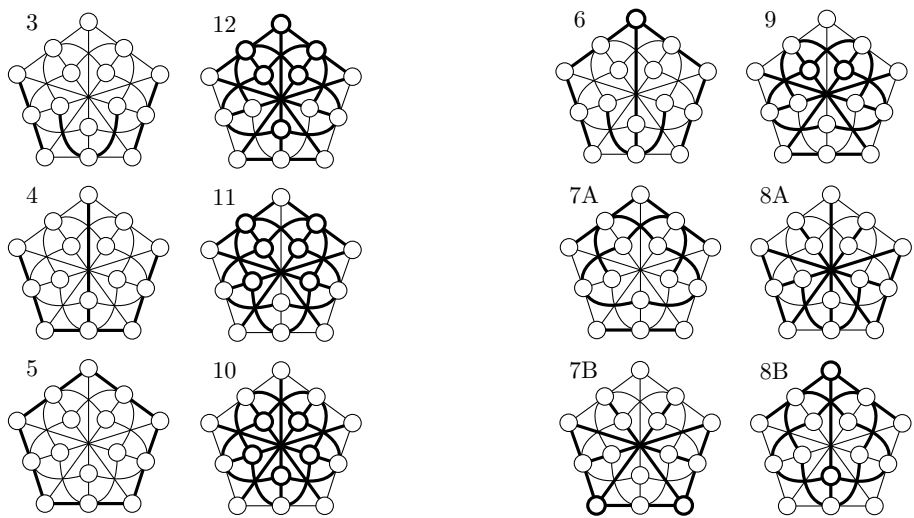

Let us give a brief description of the individual types of configurations. In
Type 3 the three negative lines are pairwise disjoint and lie in a grid; that
is, their dual is a tricentric triad. Type 4 features three pairwise disjoint
lines not belonging to a grid and their unique transversal. In Type 5 the five
negative lines form a pentagon. Type 6 contains the three lines from Type 3
plus three concurrent lines, whose point of concurrence is not lying on any of
the three former lines. Type 7A contains six lines forming a hexagon and a
unique line disjoint from any of the six. Type 7B is a particular union of two
Types 4 and an extra line or, equivalently, is composed of the five lines of a
grid and two disjoint lines. A two-qubit doily features just a Type 3 pattern,
while in a three-qubit doily we can find all the patterns from Type 3 to Type
7A inclusive~\cite{SdHG21}.

\subsection{Linear and quadratic doilies}

Following~\cite{SdHG21}, we will also distinguish between two kinds of doilies,
referred to as linear and quadratic. A \emph{linear} $N$-qubit doily spans a
PG$(3,2)$ of the ambient PG$(2N-1, 2)$. This means that the three lines of a
perp-set of such a doily are coplanar, i.\,e. lie in a PG$(2,2)$ of the
PG$(3,2)$, a tricentric triad corresponds to a line of the PG$(3,2)$ and the
plane defined by a unicentric triad of the doily passes through its center.
Figure \ref{fig:proof-linear} serves as a graphical illustration of these
features for $N=4$. In the doily we selected a perp-set (blue) and colored the
remaining lines red. The model of PG$(3,2)$ is based on a 3-D tetrahedral model
of Polster \cite{Pol98}; our version features all the points but not all the
lines of the model in order to avoid too crowded appearance of the figure. The
two red points at the side lie on the line passing via $IYZI$ that would be
perpendicular to the plane of the drawings. Each black line of the PG$(3,2)$ is
non-isotropic and corresponds to a tricentric triad in the doily.

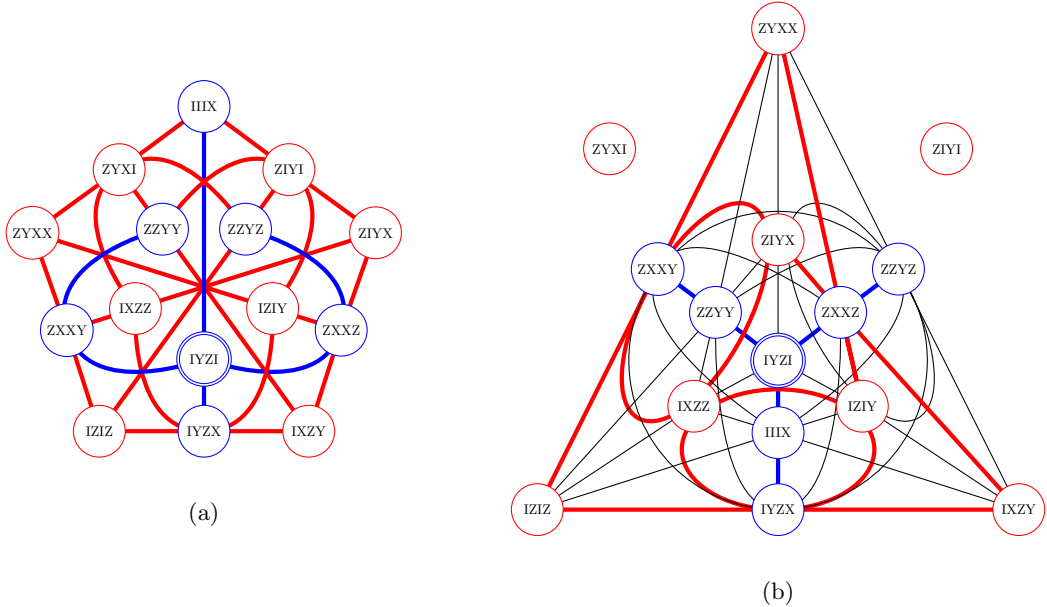
\begin{figure}[!ht]
  \begin{subfigure}{0.5\textwidth}
  \begin{center}

\begin{tikzpicture}[every plot/.style={smooth, tension=2},
  scale=2.4,
  every node/.style={scale=0.53,fill=white,circle,draw=black, minimum size=1.3cm}
]
\coordinate (iiix) at (0,1.0);
\coordinate (iyzx) at (0,-0.80);
\coordinate (iyzi) at (0,-0.40);
\coordinate (zyxx) at (-0.95,0.30);
\coordinate (zxxz) at (0.76,-0.24);
\coordinate (iziy) at (0.38,-0.12);
\coordinate (iziz) at (-0.58,-0.80);
\coordinate (ziyi) at (0.47,0.65);
\coordinate (zzyz) at (0.23,0.32);
\coordinate (ixzy) at (0.58,-0.80);
\coordinate (zyxi) at (-0.47,0.65);
\coordinate (zzyy) at (-0.23,0.32);
\coordinate (ziyx) at (0.95,0.30);
\coordinate (zxxy) at (-0.76,-0.24);
\coordinate (ixzz) at (-0.38,-0.12);
\draw [draw=red,ultra thick] (iiix) -- (zyxi) -- (zyxx);
\draw [draw=red,ultra thick] (zyxx) -- (zxxy) -- (iziz);
\draw [draw=red,ultra thick] (iziz) -- (iyzx) -- (ixzy);
\draw [draw=red,ultra thick] (ixzy) -- (zxxz) -- (ziyx);
\draw [draw=red,ultra thick] (ziyx) -- (ziyi) -- (iiix);
\draw [draw=red,ultra thick] (iziz) -- (zzyz) -- (ziyi);
\draw [draw=red,ultra thick] (zyxi) -- (zzyy) -- (ixzy);
\draw [draw=red,ultra thick] (ziyx) -- (ixzz) -- (zxxy);
\draw [draw=red,ultra thick] (zyxx) -- (iziy) -- (zxxz);
\draw [draw=blue,ultra thick] (iiix) -- (iyzi) -- (iyzx);
\draw [draw=red,ultra thick] plot coordinates{(ixzz) (zyxi) (zzyz)};
\draw [draw=red,ultra thick] plot coordinates{(zzyy) (ziyi) (iziy)};
\draw [draw=blue,ultra thick] plot coordinates{(zzyz) (zxxz) (iyzi)};
\draw [draw=red,ultra thick] plot coordinates{(iziy) (iyzx) (ixzz)};
\draw [draw=blue,ultra thick] plot coordinates{(iyzi) (zxxy) (zzyy)};
\node [draw=blue] at (iiix) {IIIX};
\node [draw=blue] at (iyzx) {IYZX};
\node [double, draw=blue] at (iyzi) {IYZI};
\node [draw=red] at (zyxx) {ZYXX};
\node [draw=blue] at (zxxz) {ZXXZ};
\node [draw=red] at (iziy) {IZIY};
\node [draw=red] at (iziz) {IZIZ};
\node [draw=red] at (ziyi) {ZIYI};
\node [draw=blue] at (zzyz) {ZZYZ};
\node [draw=red] at (ixzy) {IXZY};
\node [draw=red] at (zyxi) {ZYXI};
\node [draw=blue] at (zzyy) {ZZYY};
\node [draw=red] at (ziyx) {ZIYX};
\node [draw=blue] at (zxxy) {ZXXY};
\node [draw=red] at (ixzz) {IXZZ};
\end{tikzpicture}   \end{center}
  \caption{}
  \label{doily-linear}
  \end{subfigure}
  \begin{subfigure}{0.5\textwidth} 
  \begin{center}
\begin{tikzpicture}[every plot/.style={smooth, tension=2},
  scale=3.2,
  every node/.style={scale=0.53,fill=white,circle,draw=black, minimum size=1.3cm}
]
\coordinate (iyzi) at (0,-0.38);
\coordinate (iziz) at (-1,-1);
\coordinate (ixzy) at (1,-1);
\coordinate (iyzx) at (0,-1);
\coordinate (zzyz) at (0.5,0);
\coordinate (zxxy) at (-0.5,0);
\coordinate (zyxx) at (0,1);
\coordinate (ziyx) at (0,0.122);
\coordinate (zxxz) at (0.26,-0.18);
\coordinate (iziy) at (0.35,-0.57);
\coordinate (iiix) at (0,-0.68);
\coordinate (ixzz) at (-0.35,-0.57);
\coordinate (zzyy) at (-0.26,-0.18);
\coordinate (zyxi) at (-0.7,0.5);
\coordinate (ziyi) at (0.7,0.5);

\draw (zyxx) -- (ziyx) -- (iyzi);
\draw [draw=blue,ultra thick](iyzi) -- (iiix) -- (iyzx);
\draw [draw=red,ultra thick](zyxx) -- (zxxy) -- (iziz);
\draw (zyxx) -- (zzyz) -- (ixzy);
\draw [draw=red,ultra thick](iziz) -- (iyzx) -- (ixzy);
\draw (iziz) -- (ixzz) -- (iyzi);
\draw [draw=blue,ultra thick](iyzi) -- (zxxz) -- (zzyz);
\draw [draw=blue,ultra thick](zxxy) -- (zzyy) -- (iyzi);
\draw (iyzi) -- (iziy) -- (ixzy);
\draw[rotate=28] (0.22, -0.32) ellipse (0.22cm and 0.5cm);
\draw[rotate=-28,draw=red,ultra thick] (-0.22, -0.32) ellipse (0.22cm and 0.5cm);
\draw[draw=red,ultra thick] (0, -0.75) ellipse (0.4cm and 0.25cm);
\draw (iyzi) circle[radius=0.62];
\draw plot coordinates{(zzyy) (iyzx) (zxxz)};
\draw plot coordinates{(iiix) (zxxy) (zxxz)};
\draw plot coordinates{(zzyy) (zzyz) (iiix)};
\draw (zyxx) -- (zzyy) -- (ixzz);
\draw[draw=red,ultra thick] (zyxx) -- (iziy) -- (zxxz);
\draw[draw=red,ultra thick] (ixzy) -- (zxxz) -- (ziyx);
\draw (ixzy) -- (iiix) -- (ixzz);
\draw (iziz) -- (iiix) -- (iziy);
\draw (ziyx) -- (zzyy) -- (iziz);

\node[double,draw=blue] at (iyzi) {IYZI};%
\node[draw=red] at (iziz) {IZIZ};%
\node[draw=red] at (ixzy) {IXZY};%
\node[draw=blue] at (iyzx) {IYZX};%
\node[draw=blue] at (zzyz) {ZZYZ};%
\node[draw=blue] at (zxxy) {ZXXY};%
\node[draw=red] at (zyxx) {ZYXX};%
\node[draw=red] at (ziyx) {ZIYX};%
\node[draw=blue] at (zxxz) {ZXXZ};%
\node[draw=red] at (iziy) {IZIY};%
\node[draw=blue] at (iiix) {IIIX};%
\node[draw=red] at (ixzz) {IXZZ};%
\node[draw=blue] at (zzyy) {ZZYY};%
\node[draw=red] at (zyxi) {ZYXI};%
\node[draw=red] at (ziyi) {ZIYI};%
\end{tikzpicture}   \end{center}
  \caption{}
  \label{fano-4space}
  \end{subfigure}
  \caption{A linear four-qubit doily with one of its perp-sets highlighted in 
    blue color (a) and the corresponding PG$(3,2)$ of PG$(7,2)$ it spans (b). One 
    can readily see that the three lines of the perp-set lie in a plane of the 
    PG$(3,2)$ and the three points on a non-isotropic line of the space (black) 
    correspond to a tricentric triad of the doily.}
  \label{fig:proof-linear}
\end{figure}

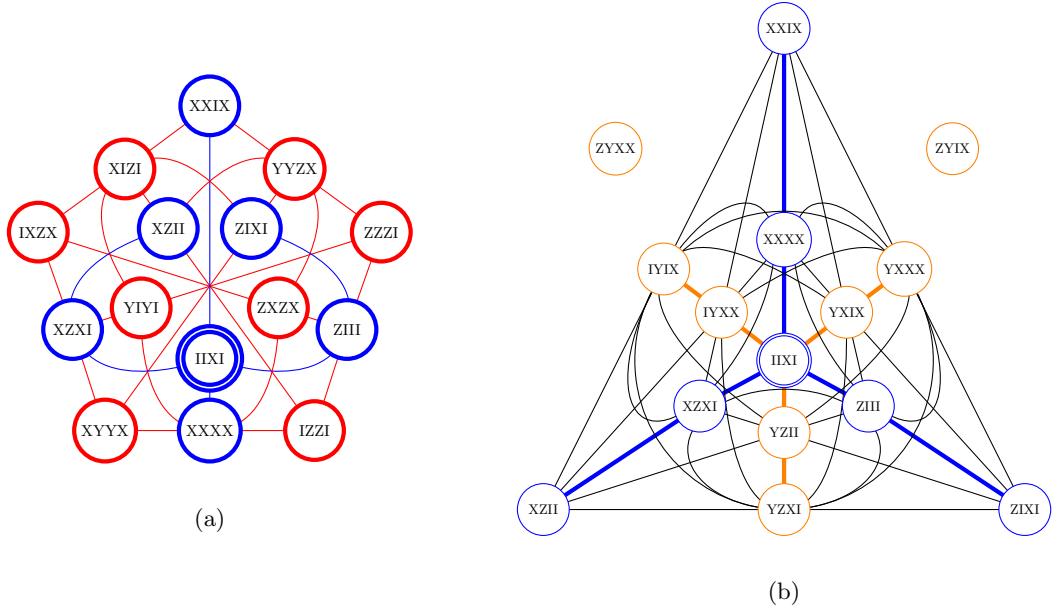
\begin{figure}[!ht]
  \begin{subfigure}{0.5\textwidth}
  \begin{center}
\begin{tikzpicture}[every plot/.style={smooth, tension=2},
  scale=2.4,
  every node/.style={scale=0.6,fill=white,circle,draw=black, minimum size=1.3cm}
]
\coordinate (xxix) at (0,1.0);
\coordinate (xxxx) at (0,-0.80);
\coordinate (iixi) at (0,-0.40);
\coordinate (ixzx) at (-0.95,0.30);
\coordinate (ziii) at (0.76,-0.24);
\coordinate (zxzx) at (0.38,-0.12);
\coordinate (xyyx) at (-0.58,-0.80);
\coordinate (yyzx) at (0.47,0.65);
\coordinate (zixi) at (0.23,0.32);
\coordinate (izzi) at (0.58,-0.80);
\coordinate (xizi) at (-0.47,0.65);
\coordinate (xzii) at (-0.23,0.32);
\coordinate (zzzi) at (0.95,0.30);
\coordinate (xzxi) at (-0.76,-0.24);
\coordinate (yiyi) at (-0.38,-0.12);
\draw [draw=red] (xxix) -- (xizi) -- (ixzx);
\draw [draw=red] (ixzx) -- (xzxi) -- (xyyx);
\draw [draw=red] (xyyx) -- (xxxx) -- (izzi);
\draw [draw=red] (izzi) -- (ziii) -- (zzzi);
\draw [draw=red] (zzzi) -- (yyzx) -- (xxix);
\draw [draw=red] (xyyx) -- (zixi) -- (yyzx);
\draw [draw=red] (xizi) -- (xzii) -- (izzi);
\draw [draw=red] (zzzi) -- (yiyi) -- (xzxi);
\draw [draw=red] (ixzx) -- (zxzx) -- (ziii);
\draw [draw=blue] (xxix) -- (iixi) -- (xxxx);
\draw [draw=red] plot coordinates{(yiyi) (xizi) (zixi)};
\draw [draw=red] plot coordinates{(xzii) (yyzx) (zxzx)};
\draw [draw=blue] plot coordinates{(zixi) (ziii) (iixi)};
\draw [draw=red] plot coordinates{(zxzx) (xxxx) (yiyi)};
\draw [draw=blue] plot coordinates{(iixi) (xzxi) (xzii)};

\node [ultra thick,draw=blue] at (xxix) {XXIX};
\node [ultra thick,draw=blue] at (xxxx) {XXXX};
\node [double, ultra thick,draw=blue] at (iixi) {IIXI};
\node [ultra thick,draw=red] at (ixzx) {IXZX};
\node [ultra thick,draw=blue] at (ziii) {ZIII};
\node [ultra thick,draw=red] at (zxzx) {ZXZX};
\node [ultra thick,draw=red] at (xyyx) {XYYX};
\node [ultra thick,draw=red] at (yyzx) {YYZX};
\node [ultra thick,draw=blue] at (zixi) {ZIXI};
\node [ultra thick,draw=red] at (izzi) {IZZI};
\node [ultra thick,draw=red] at (xizi) {XIZI};
\node [ultra thick,draw=blue] at (xzii) {XZII};
\node [ultra thick,draw=red] at (zzzi) {ZZZI};
\node [ultra thick,draw=blue] at (xzxi) {XZXI};
\node [ultra thick,draw=red] at (yiyi) {YIYI};
\end{tikzpicture}   \end{center}
  \caption{}
  \label{doily-quadratic}
  \end{subfigure}
  \begin{subfigure}{0.5\textwidth} 
  \begin{center}
\begin{tikzpicture}[every plot/.style={smooth, tension=2},
  scale=3.2,
  every node/.style={scale=0.53,fill=white,circle,draw=black, minimum size=1.3cm}
]
\coordinate (iixi) at (0,-0.38);
\coordinate (xzii) at (-1,-1);
\coordinate (zixi) at (1,-1);
\coordinate (yzxi) at (0,-1);
\coordinate (yxxx) at (0.5,0);
\coordinate (iyix) at (-0.5,0);
\coordinate (xxix) at (0,1);
\coordinate (xxxx) at (0,0.122);
\coordinate (yxix) at (0.26,-0.18);
\coordinate (ziii) at (0.35,-0.57);
\coordinate (yzii) at (0,-0.68);
\coordinate (xzxi) at (-0.35,-0.57);
\coordinate (iyxx) at (-0.26,-0.18);
\coordinate (zyxx) at (-0.7,0.5);
\coordinate (zyix) at (0.7,0.5);

\draw [draw=blue,ultra thick] (xxix) -- (xxxx) -- (iixi);
\draw [draw=orange,ultra thick](iixi) -- (yzii) -- (yzxi);
\draw (xzii) -- (iyix) -- (xxix);
\draw (xxix) -- (yxxx) -- (zixi);
\draw (xzii) -- (yzxi) -- (zixi);
\draw [draw=blue,ultra thick](xzii) -- (xzxi) -- (iixi);
\draw [draw=orange,ultra thick](iixi) -- (yxix) -- (yxxx);
\draw [draw=orange,ultra thick](iyix) -- (iyxx) -- (iixi);
\draw [draw=blue,ultra thick](iixi) -- (ziii) -- (zixi);
\draw[rotate=28] (0.22, -0.32) ellipse (0.22cm and 0.5cm);
\draw[rotate=-28] (-0.22, -0.32) ellipse (0.22cm and 0.5cm);
\draw (0, -0.75) ellipse (0.4cm and 0.25cm);
\draw (iixi) circle[radius=0.62];
\draw plot coordinates{(iyxx) (yzxi) (yxix)};
\draw plot coordinates{(yzii) (iyix) (yxix)};
\draw plot coordinates{(iyxx) (yxxx) (yzii)};
\draw plot [tension=1] coordinates{(xxix) (iyxx) (xzxi)};
\draw (xxix) -- (yxix) -- (ziii);
\draw plot [tension=1] coordinates{(zixi) (yxix) (xxxx)};
\draw (zixi) -- (yzii) -- (xzxi);
\draw (xzii) -- (yzii) -- (ziii);
\draw (xxxx) -- (iyxx) -- (xzii);

\node[double,draw=blue] at (iixi) {IIXI};%
\node[draw=blue] at (xzii) {XZII};%
\node[draw=blue] at (zixi) {ZIXI};%
\node[draw=orange] at (yzxi) {YZXI};%
\node[draw=orange] at (yxxx) {YXXX};%
\node[draw=orange] at (iyix) {IYIX};%
\node[draw=blue] at (xxix) {XXIX};%
\node[draw=blue] at (xxxx) {XXXX};%
\node[draw=orange] at (yxix) {YXIX};%
\node[draw=blue] at (ziii) {ZIII};%
\node[draw=orange] at (yzii) {YZII};%
\node[draw=blue] at (xzxi) {XZXI};%
\node[draw=orange] at (iyxx) {IYXX};%
\node[draw=orange] at (zyxx) {ZYXX};%
\node[draw=orange] at (zyix) {ZYIX};%
\end{tikzpicture}   \end{center}
  \caption{}
  \label{fano-3space}
  \end{subfigure}
  \caption{An illustration of the fact that a perp-set of a quadratic 
    (four-qubit) doily (a) spans a PG$(3,2)$ (b) of the PG$(4,2)$ spanned by the
    doily.}
  \label{fig:proof-quadratic}
\end{figure}

A quadratic $N$-qubit doily spans a PG$(4,2)$ of the ambient PG$(2N-1,2)$, being,
in fact, isomorphic to the geometry formed by 15 points and 15 lines lying on a
parabolic quadric $Q(4,2)$ in this PG$(4,2)$. This quadric, as any other
parabolic quadric in PG$(4,2)$, has a remarkable property that all its tangent
hyperplanes pass through the same point $J$, called the nucleus (see, e.\,g.,
\cite{HT91}). Any tricentric triad of such a doily defines a plane in the 
PG$(4,2)$ that contains $J$; a unicentric triad also defines a plane, this plane
passing through the remaining third point lying on the line defined by $J$ and
the (unique) center of the triad. Moreover, all the 15 PG$(3,2)$s passing
through $J$ intersect our quadric in three concurrent lines that form a
perp-set of the doily. Figure \ref{fig:proof-quadratic} offers a pictorial
illustration of some of these properties. We again take a four-qubit doily,
where we highlighted a perp-set (blue). Now the three lines of the perp-set are
not coplanar as in the case of linear doily, but span a PG$(3,2)$. We colored
the remaining eight points (and the totally-isotropic lines) of the PG$(3,2)$ in
yellow in order to stress the property that the only points shared by the doily
and this PG$(3,2)$ are the (blue) points of the perp-set. There are two
``distinguished'' points of the PG$(3,2)$, namely $ZYXX$ and $ZYIX$, which lie on
the remaining seventh line passing via $IIXI$; the point $ZYXX$ is nothing but
the nucleus of the parabolic quadric our particular doily is located on. Given a
perp-set, we know that there are four tricentric and four unicentric triads
contained in it. In our particular perp, the four \emph{tri}centric triads are
$\{XXIX,XZII,ZIXI\}$, $\{XXIX,XZXI,ZIII\}$, $\{XZXI,XXXX,ZIXI\}$ and 
$\{XXXX,ZIII,XZII\}$; one can readily check that the product of the three
observables in any of them is $ZYXX$ (the nucleus). The four \emph{uni}centric 
triads of our perp-set are $\{XZXI,XXXX,ZIII\}$, $\{XZXI,ZIXI,XXIX\}$, 
$\{ZIII,XZII,XXIX\}$ and $\{XXX,XZII,ZIXI\}$; the product of the observables in
any of them is $ZYIX$, i.\,e. the second distinguished point. By this
construction we get a (different) PG$(3,2)$ for any of the 15 perp-sets of the
doily; and because in any of these perp-sets the four tricentric triads always
define the nucleus, $ZYXX$, we get altogether 15 PG$(3,2)$s that share the point
$ZYXX$, these 15 spaces lying in that particular PG$(4,2)$ of the ambient PG$
(7,2)$ that contains the quadric of our selected doily.

In a recent paper~\cite{SdHG21}, four of the authors
have thoroughly analyzed and classified three-qubit doilies. To this end, they
first explicitly computed all 63 perp-sets, 36 hyperbolic quadrics and 28
elliptic quadrics living in $\mathcal{W}(5,2)$. Then, employing the fact that a linear doily
is isomorphic to the intersection of two perp-sets with non-collinear nuclei,
they computed and classified all $63 \times 32/3! = 336$ linear doilies of the
$\mathcal{W}(5,2)$.
In the next step, making use of the property that a quadratic doily is
isomorphic to the intersection of an elliptic quadric and a hyperbolic quadric,
they generated and classified all $36 \times 28 = \numprint{1008}$ quadratic doilies of the
$\mathcal{W}(5,2)$.
The procedure described above is, however, not a viable one for $N > 3$, as we
would first need to compute all $\mathcal{W}(5,2)$s living in a particular $\mathcal{W}(2N-1,2)$, $N>3$, and
then in each of them compute 336 linear and \numprint{1008} quadratic doilies following the
strategy of \cite{SdHG21}. Instead, we shall follow (in Section 4) a different,
and reasonably faster, approach that makes use of some properties of an ovoid of
a doily. In particular, we shall start with a particular $N$-qubit ovoid, i.\,e.
a set of five mutually anticommuting $N$-qubit observables whose product is
$\pm{\cal I}_N$, and introduce a unique algebro-geometrical recipe with the help
of which one can find all the $N$-qubit doilies having this particular ovoid in
common. Before embarking on this path, however, we shall introduce several
general formulas for the number of both linear and quadratic doilies of $\mathcal{W}(2N-1,2)$,
valid for any $N \geq 2$, so that we already have certain important numbers at
hand to validate some of our subsequent, mostly computer-assisted,
results.

\subsection{Contextuality degree}

All multi-qubit doilies are observable-based proofs of the Kochen--Specker
theorem, that establishes that no Non-Contextual Hidden Variables (NCHV) model
can reproduce the outcomes of quantum mechanics. This contextuality property is
 related to a linear problem, as follows.  Let $A$ be the
incidence matrix of the points on the lines of a finite geometry, such as the
doily. Its coefficients are in the two-elements field $\mathbb{F}_2 =
\{0,1\}$, its $l$ rows correspond to the
geometric lines and its $p$ columns to the geometric points (for the doily,
$l=p=15$). The positive (resp. negative) nature of a line is encoded by a $0$
(resp. $1$) for the corresponding coefficient of the \emph{valuation vector} $E$
in $\mathbb{F}_2^{l}$. Then a quantum geometry is contextual iff there is no
vector $x$ such that $Ax = E$. The \emph{contextuality degree} is the minimal
Hamming distance between a vector $Ax$ and the vector $E$~\cite{dHG+22}. The
contextuality degree is the minimal number of line valuations
that one should change to make the quantum geometry satisfiable by an NCHV
model.

\begin{proposition}\label{contextuality-prop}
All multi-qubit doilies have a contextuality degree of 3.
\end{proposition}

\begin{proof}
All multi-qubit doilies have the same incidence matrix $A$. Accordingly, the
only parameter that is changing between all the doilies is the vector $E$, which
only depends on the configuration of their negative lines. We have seen that
there are only 12 such configurations. For each of these 12 configurations, we
have computed the Hamming distance between $Ax$ and $E$, for all vectors $x$ in
$\mathbb{F}_2^{15}$. It turns out that the minimal Hamming distance is always 3.
\end{proof}

In practice, we did not write by hand the 12 possible $E$ vectors, but we
computed these vectors from the 5-qubit doilies, because, as described later,
we have checked by enumeration that these doilies present all the configurations.

\section{Numbers of multi-qubit doilies}
\label{nb-sec}

This section proposes and justifies closed formulas for the numbers 
 of linear and quadratic doilies in $\mathcal{W}(2N-1,2)$.
Before all we introduce some well-known formulas. First, we introduce the
 Gaussian (binomial) coefficient 
\begin{equation}
\qbin{n}{k} = \prod_{i=1}^{k} \frac{q^{n-k+i} -1}{q^i -1} 
            = \frac{(q^n-1) \dots (q^{n-k+1} -1)}{(q^k-1) \dots (q-1)}
\label{gauss}
\end{equation}
where $0 \leq k \leq n$ and $q$ is a power of a prime, which gives the number of
subspaces of dimension\footnote{All dimensions in this section are projective 
dimensions.} $k-1$ in a projective space PG$(n-1,q)$ of dimension $n-1$ over 
$\mathbb{F}_q$. More generally, the number of  $(k-1)$-dimensional spaces of PG$(n-1,q)$
that pass through a fixed $(l-1)$-dimensional space is
\begin{equation}
\qbin{n-l}{k-l}.
\label{subvsub-pr}
\end{equation}
Next, for a symplectic polar space $\mathcal{W}(2N-1,q)$ embedded in a projective 
space PG$(2N-1,q)$, the number of its $k$-dimensional spaces is given by (see,
e.\,g.,~\cite[Lemma 2.10]{DRSS19})
\begin{equation}
  \qbin{N}{k+1}  \prod_{i=1}^{k+1} \left(q^{N+1-i} +1\right)
  \label{sub-sympl}
\end{equation}
and the number of $k$-dimensional spaces through a fixed $m$-dimensional 
space~\cite[Corollary 2.11]{DRSS19} equals 
\begin{equation}
\qbin{N-m-1}{k-m}  \prod_{i=1}^{k-m} \left(q^{N-m-i} +1\right).
\label{subvsub-sympl}
\end{equation}
Further, let $\perp$ be a symplectic polarity of PG$(n,q)$ and let denote by
$S^{\perp}$ the polar space of a subspace $S$. 
If $S$ is of dimension $k$, then $S^{\perp}$ has dimension $n-k-1$.
A projective subspace $S$ of
PG$(n,q)$ is called \emph{isotropic} if $S \cap S^{\perp} \neq \emptyset$ and
\emph{non-isotropic} if $S \cap S^{\perp}=\emptyset$. An isotropic $S$ is called
\emph{totally isotropic} if $S \subseteq S^{\perp}$. It is easy to see that if
$S$ is a totally isotropic subspace, then every subspace contained in $S$ is
also totally isotropic. Moreover,
\begin{equation}
S \subseteq T^{\perp} \Rightarrow T \subseteq S^{\perp}.
\label{sub}
\end{equation}
In order to prove the two theorems below, we will need a couple of lemmas.
\begin{lemma}
\label{2via1}
If a PG$(3,2)$  of the ambient PG$(5,2)$ equipped with a symplectic 
polarity $\perp$ contains a totally-isotropic PG$(2,2)$, then it contains
exactly three such PG$(2,2)$s, passing through a common (totally-isotropic) 
PG$(1,2)$.
\end{lemma}
\begin{proof}
First, there are no totally-isotropic PG$(3,2)$s in the PG$(5,2)$. Given a
totally-isotropic PG$(1,2)$ of PG$(5,2)$, $S$, there are (see 
Eq.\,(\ref{subvsub-sympl}) for $q=2$, $N=3$, $k=2$ and $m=1$) three
totally-isotropic PG$(2,2)$s passing through it. Denoting these as 
$T^{\perp}_{i}$ ($i = 1, 2, 3$), the lemma then follows from the fact that 
$S^{\perp} \cong$ PG$(3,2)$, $T^{\perp}_{i} = T_i$, and property (\ref{sub}).
\end{proof}
\begin{remark}
\label{rm52}
For $N > 3$, PG$(2N-1,2)$ features also totally-isotropic PG$(3,2)$s; any other 
of its PG$(3,2)$s endowed with totally-isotropic PG$(2,2)$s has the property as 
described in Lemma \ref{2via1}.
\end{remark}
\begin{lemma}
\label{3via2}
If a PG$(4,2)$ of the ambient PG$(7,2)$ equipped with a symplectic 
polarity $\perp$ contains a totally-isotropic PG$(3,2)$, then it contains
exactly three such PG$(3,2)$s, passing through a common (totally-isotropic) 
PG$(2,2)$.
\end{lemma}
\begin{proof}
The proof parallels that of the preceding lemma. First, there are no
totally-isotropic PG$(4,2)$s in the PG$(7,2)$. Given a totally-isotropic 
PG$(2,2)$ of PG$(7,2)$, $S$, there are (see Eq.\,(\ref{subvsub-sympl}) for $q=2$, 
$N=4$, $k=3$ and $m=2$) three %
totally-isotropic PG$(3,2)$s passing through it. Denoting these as 
$T^{\perp}_{i}$ ($i = 1, 2, 3$), the lemma then follows from the fact that 
$S^{\perp} \cong$ PG$(4,2)$, $T^{\perp}_{i} = T_i$, and property (\ref{sub}).
\end{proof}
\begin{remark}
\label{rm72}
For $N > 4$, PG$(2N-1,2)$ features also totally-isotropic PG$(4,2)$s; any other 
of its PG$(4,2)$s endowed with totally-isotropic PG$(3,2)$s has the property as 
described in Lemma \ref{3via2}.
\end{remark}

Next, through a (totally-isotropic) point of PG$(5,2)$, $S$, there pass 15
totally-isotropic PG$(1,2)$s, $T_j^{\perp}$ ($j=1,2,3, \dots,15$) and the same
number of PG$(2,2)$s. Given the facts that $S^{\perp} \cong$ PG$(4,2)$ and
$T_j \cong$ PG$(3,2)$, a PG$(4,2)$ of PG$(5,2)$ will contain 15 PG$(3,2)$s of
type defined by Lemma~\ref{2via1} concurring at a point, namely the pole of
this particular PG$(4,2)$. As PG$(4,2)$ contains altogether 31 PG$(3,2)$s, each
of the remaining 16 PG$(3,2)$s does not contain totally-isotropic PG$(2,2)$s and
so hosts a unique linear doily. As each such doily can be viewed as a projection
of a quadratic doily from the pole, a PG$(4,2)$ is found to be spanned by 16
quadratic doilies.

\begin{remark}
\label{rm16d}
If a PG$(4,2)$ of the ambient  PG$(2N-1,2)$, $N > 3$, is devoid of 
totally-isotropic PG$(3,2)$s, then it is of the type described above, i.\,e. it 
entails 16 quadratic doilies.
\end{remark}

\subsection{Number of linear doilies}

\begin{theorem}
For any $N \geq 2$ the number of linear doilies in $\mathcal{W}(2N-1,2)$ is
\begin{equation}
D_l(N) = \qbincoeff{2N}{4}{2}
-
\qbincoeff{N}{4}{2} ~\prod_{i=1}^{4} \left(2^{N+1-i} +1\right)
-
7 \qbincoeff{N}{3}{2} 2^{2N-6} ~\prod_{i=1}^{3} \left(2^{N+1-i} +1\right)
/3.
\label{lN}
\end{equation}
\end{theorem}

\begin{proof}
A linear doily of $\mathcal{W}(2N-1,2)$ spans a particular PG$(3,2)$ of the
ambient PG$(2N-1,2)$ that does not contain any totally-isotropic PG$(2,2)$. And
since any such PG$(3,2)$ is spanned by a single linear doily, the number of
linear doilies of $\mathcal{W}(2N-1,2)$ is thus equal to the number of PG$(3,2)$s 
that are devoid of totally-isotropic planes. To find this number, from
Eq.\,(\ref{gauss}) we first note that
there are altogether
\begin{equation}
\qbincoeff{2N}{4}{2}
\label{lan}
\end{equation}
PG$(3,2)$s in PG$(2N-1,2)$, out of which 
\begin{equation}
\qbincoeff{N}{4}{2} ~\prod_{i=1}^{4} \left(2^{N+1-i} +1\right)
\label{tin}
\end{equation}
(Eq.\,(\ref{sub-sympl}) with $k=3$ and $q=2$) are totally isotropic.

To ascertain the cardinality of the remaining PG$(3,2)$s that feature 
totally-isotropic PG$(2,2)$s, we proceed as follows. We first observe that by
Eq.\,(\ref{sub-sympl}) with $k=2$ and $q=2$ there are 
\begin{equation} %
\qbincoeff{N}{3}{2} ~\prod_{i=1}^{3} \left(2^{N+1-i} +1\right)
\label{tipn}
\end{equation}
totally-isotropic PG$(2,2)$s 
in PG$(2N-1,2)$. Next, with $k=3$ and $m=2$ in (\ref{subvsub-sympl}), it follows 
that there are 
\begin{equation}%
   \qbincoeff{N-3}{1}{2} \left(2^{N-3} +1\right)
  = 2^{2(N-3)}-1
\label{lptipn}
\end{equation}
\noindent totally-isotropic PG$(3,2)$s passing through a totally-isotropic 
PG$(2,2)$. And since the total number of PG$(3,2)$s passing via a PG$(2,2)$ of
PG$(2N-1,2)$ is
\begin{equation}%
\qbincoeff{2N-3}{4-3}{2} 
  = 2^{2N-3}-1
\label{lppn}
\end{equation}
(as stemming from Eq.\,(\ref{subvsub-pr}) for $n=2N$, $k=4$, $l=3$ and $q=2$),
through a totally-isotropic PG$(2,2)$ there pass 
\begin{equation}
2^{2N-3}-1 - \left(2^{2(N-3)}-1\right)
= 7~\times~2^{2N-6}
\label{lnptipn}
\end{equation}
isotropic PG$(3,2)$s apart from those that are totally isotropic. Hence, the
number of those PG$(3,2)$s of PG$(2N-1,2)$ that are endowed with
totally-isotropic PG$(2,2)$s -- with the exclusion of totally isotropic ones --
amounts to 
$(\ref{tipn}) \times (\ref{lnptipn})/3$,
where we also took into account (see Remark \ref{rm52}) that any such PG$(3,2)$
features just three totally-isotropic PG$(2,2)$s. All in all, there are 
\begin{equation*}
\qbincoeff{2N}{4}{2}
-
\qbincoeff{N}{4}{2} ~\prod_{i=1}^{4} \left(2^{N+1-i} +1\right)
-
7 \qbincoeff{N}{3}{2} 2^{2N-6} ~\prod_{i=1}^{3} \left(2^{N+1-i} +1\right)
/3
\end{equation*}
PG$(3,2)$s in the ambient PG$(2N-1,2)$ that are devoid of totally-isotropic 
PG$(2,2)$s, and so the same number of
linear doilies in $\mathcal{W}(2N-1,2)$.
\end{proof}

Given the fact that the three lines of a perp-set of a linear doily span a 
PG$(2,2)$, and namely that PG$(2,2)$ that features just three totally-isotropic
PG$(1,2)$s, we arrive at the interesting expression
\begin{equation}
D_l(N) = \frac{4}{15} 4^{N-3} \Theta_2(N),
\label{dl-simp}
\end{equation}
for the number of linear doilies in $\mathcal{W}(2N-1,2)$, where
\begin{equation}
\label{theta32}
\Theta_2(N) = \frac{1}{16}2^{2N}
              \prod_{i=1}^2\frac{2^{N-2+i}-1}{2^i-1}
              \prod_{i=1}^2(2^{N+1-i} +1)
\end{equation}
is the number of those PG$(2,2)$s of the ambient PG$(2N-1,2)$ each of which 
features just three totally-isotropic PG$(1,2)$s.

\subsection{Number of quadratic doilies}

\begin{theorem}
For any $N \geq 3$ the number of quadratic doilies in $\mathcal{W}(2N-1,2)$ is
\begin{equation}
\label{qN}
D_q(N) = 16\left(\qbincoeff{2N}{5}{2}
                -\qbincoeff{N}{5}{2}\prod_{i=1}^{5}\left(2^{N+1-i} +1\right)
                -15\qbincoeff{N}{4}{2}2^{2N-8}\prod_{i=1}^{4}\left(2^{N+1-i}+1\right)/3
           \right).
\end{equation}
\end{theorem}

\begin{proof}
A quadratic doily of $\mathcal{W}(2N-1,2)$ spans a particular PG$(4,2)$ of the
ambient PG$(2N-1,2)$ that does not contain any totally-isotropic PG$(3,2)$. And
since any such PG$(4,2)$ is spanned by (see Remark \ref{rm16d}) 16 such doilies
that are all unique to this space, the number of quadratic doilies of 
$\mathcal{W}(2N-1,2)$ is thus equal to 16 times the number of PG$(4,2)$s that are
devoid of totally-isotropic PG$(3,2)$s. To find the latter number, we again
start with Eq.\,(\ref{gauss}) that tells us that there are altogether
\begin{equation}
\qbincoeff{2N}{5}{2}
\label{qan}
\end{equation}
PG$(4,2)$s in PG$(2N-1,2)$, out of which 
\begin{equation}
\qbincoeff{N}{5}{2} ~\prod_{i=1}^{5} \left(2^{N+1-i} +1\right)
\label{qtin}
\end{equation}
(Eq.\,(\ref{sub-sympl}) with $k=4$ and $q=2$) are totally isotropic.

To ascertain the cardinality of the remaining isotropic PG$(4,2)$s, we proceed as
follows. We first observe that by Eq.\,(\ref{sub-sympl}) with $k=3$ and $q=2$ 
there are 
\begin{equation}
\qbincoeff{N}{4}{2} ~\prod_{i=1}^{4} \left(2^{N+1-i} +1\right)
\label{qtipn}
\end{equation}
totally-isotropic PG$(3,2)$s 
in PG$(2N-1,2)$. Next, with $k=4$, $m=3$ and $q=2$ in (\ref{subvsub-sympl}) it 
follows that there are 
\begin{equation}
 \qbincoeff{N-4}{1}{2} \left(2^{N-4}+1\right)
= 2^{2(N-4)} - 1
\label{qptipn}
\end{equation}
\noindent totally-isotropic PG$(4,2)$s passing through a totally-isotropic 
PG$(3,2)$. And since the total number of PG$(4,2)$s passing via a PG$(3,2)$ of
PG$(2N-1,2)$ is
\begin{equation}
2^{2N-4} -1
\label{nb5}
\end{equation}
(as stemming from Eq.\,(\ref{subvsub-pr}) for $n=2N$, $k=5$, $l=4$ and $q=2$),
through a totally-isotropic PG$(3,2)$ there pass 
\begin{equation}
2^{2N-4} -1 - \left(2^{2(N-4)} - 1\right)
= 15~\times~2^{2N-8}
\label{qnptipn}
\end{equation} 
PG$(4,2)$s that feature totally-isotropic PG$(3,2)$s apart from those that are 
totally isotropic. Hence, the number of those PG$(4,2)$s of PG$(2N-1,2)$ that
contain totally-isotropic PG$(3,2)$s -- with the exclusion of totally isotropic
ones -- amounts to 
$(\ref{qtipn}) \times (\ref{qnptipn}) / 3$, 
where we also took into account (see Remark \ref{rm72}) that any such PG$(4,2)$
features just three totally-isotropic PG$(3,2)$s. All in all, there are 
\begin{equation}
\qbincoeff{2N}{5}{2}
- \qbincoeff{N}{5}{2} ~\prod_{i=1}^{5} \left(2^{N+1-i} +1\right)
- 15 \qbincoeff{N}{4}{2} 2^{2N-8} \prod_{i=1}^{4} \left(2^{N+1-i} +1\right) /3
\end{equation}
PG$(4,2)$s  in the ambient PG$(2N-1,2)$ that are not endowed with any 
totally-isotropic PG$(3,2)$s, the number of quadratic doilies of 
$\mathcal{W}(2N-1,2)$ being just 16 times this number.
\end{proof}

Employing further the fact that the three lines of a perp-set of a quadratic
doily span a PG$(3,2)$, in particular that PG$(3,2)$ that contains just three
totally-isotropic PG$(2,2)$s, we find the compact formula
\begin{equation}
D_q(N) = \frac{48}{15} 4^{N-3} \Theta_3(N)
\label{dq-simp}
\end{equation}
for the number of quadratic doilies in $\mathcal{W}(2N-1,2)$, 
where
\begin{equation}
\Theta_3(N) = \frac{7}{3} 2^{2N-6}
              \prod_{i=1}^3\frac{2^{N-3+i}-1}{2^i-1}
              \prod_{i=1}^3\left(2^{N+1-i} +1\right)
\label{theta33}
\end{equation}
is the number of those PG$(3,2)$s of the ambient PG$(2N-1,2)$ each of which 
features just three totally-isotropic PG$(2,2)$s.

\begin{table}[h!]
$$
\begin{array}{l|l|l|l}
N & D_l(N) & D_q(N) & D(N)\\
\hline
2 & \numprint{1} & \_ & \numprint{1}
\\
3 & \numprint{336} & \numprint{1008} & \numprint{1344}
\\ 
4 & \numprint{91392} & \numprint{1370880} & \numprint{1462272}
\\
5 & \numprint{23744512} & \numprint{1495904256} & \numprint{1519648768}
\\
6 & \numprint{6100942848} & \numprint{1555740426240} & \numprint{1561841369088}
\\
7 & \numprint{1563272675328} & \numprint{1599227946860544} & \numprint{1600791219535872}
\\
8 & \numprint{400289425260544} & \numprint{1639185196441927680} & \numprint{1639585485867188224}
\\
9 & \numprint{102479956839235584} & \numprint{1678929132897196572672} & \numprint{1679031612854035808256}
\end{array}
$$
\caption{First numbers $D(N)$ (resp. $D_l(N)$, $D_q(N)$) of (resp. linear, 
quadratic) $N$-qubit doilies.\label{nbdoilies-tab}}
\end{table}

Comparing expressions (\ref{lN}) and (\ref{qN}), one gets
\begin{equation}
D_q(N) = (4^{N-2} - 1) D_l(N).
\label{qvsl}
\end{equation}
Consequently, the total number of doilies is
\begin{equation}
  D(N) = 4^{N-2} D_l(N).
\label{dN}
\end{equation}
For $2\leq N \leq 9$ the numbers of $N$-qubit doilies are collected in
Table~\ref{nbdoilies-tab}. Since the quadratic doilies of $\mathcal{W}(2N-1,2)$
span a PG$(4,2)$, it has no geometrical meaning to consider quadratic doilies
for $N = 2$. It can nevertheless be noticed that Eq.~(\ref{qN}) also holds
for $N = 2$, and consistently gives $D_q(2) = 0$.

\section{Generation of all \texorpdfstring{$N$}{N}-qubit doilies}
\label{sec:enum-sec}

An $N$-qubit doily can be represented by an isomorphism $f$, sometimes called a
(\emph{doily}) \emph{labeling}, mapping the points of $W_2$ to distinct points
of $W_N$, preserving commutations and anticommutations, and such that $f(a.b) = 
\pm f(a).f(b)$ for any two commuting points/observables $a$ and $b$ (the dot 
($.$) denotes the matrix product). This section describes an algorithm for the 
enumeration of all $N$-qubit doilies, for any $N \geq 2$, by construction of one 
of their labelings.

Let us start with some definitions. An $N$\emph{-qubit ovoid} is a 5-set of
mutually anticommuting $N$-qubit observables whose product is the identity
${\cal I}_N$. A \emph{triad} is a 3-set of mutually anticommuting $N$-qubit
observables. A \emph{center} of a triad is a point commuting with the three
points of the triad. A \emph{unicentric triad} is a triad that has only one
center. Let $\varepsilon$ denote the empty word. The \emph{lexicographic order}
$<$ on words is such that $\varepsilon < u$ for all non-empty word $u$, and $a.u
< b.v$ if and only if either $a < b$, or $a = b$ and $u < v$, for any letters
$a$ and $b$ and words $u$ and $v$.

In order to avoid to consider several times objects that are similar but
differently ordered, we define as follows a total order among letters and words,
and then extend it to all tuples and sets of objects of the same nature, such as
lines, sets of lines, etc. Pauli observables, encoded as words on the alphabet
$\{I,X,Y,Z\}$, are totally ordered by the lexicographic order $<$ induced by the
order on letters, also denoted $<$, such that $I < X < Z < Y$. These orders are
chosen so that their binary counterpart through the encoding $I \rightarrow 00$, $X
\rightarrow 01$, $Z \rightarrow 10$, $Y \rightarrow 11$ is the lexicographic
order on bit vectors (aka. bytes or binary words) induced by the order $0 < 1$
on bits. This order $<$ extends further to tuples $(a_1,
a_2, \ldots, a_n)$ of words, by considering them as words $a_1 a_2 \cdots a_n$
and re-using the former lexicographic order on words.
It also extends to sets of words, by associating canonically to each set the
tuple $(a_1, a_2, \ldots, a_n)$ of its elements written in increasing order
($a_i < a_j$ when $ i < j$), and so on at any level of the hierarchy of objects
of the same nature, such as a point-line geometries, seen as sets of lines, that
are sets of points.

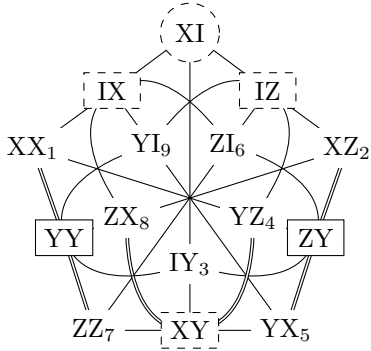
\begin{figure}[!ht]
\begin{center}
\begin{tikzpicture}[every plot/.style={smooth, tension=2},
  scale=2.2,
  every node/.style={fill=white,minimum width=0.75cm}
]
\coordinate (xi) at (0,1.0);
\coordinate (xy) at (0,-0.80);
\coordinate (iy) at (0,-0.40);
\coordinate (xx) at (-0.95,0.30);
\coordinate (zy) at (0.76,-0.24);
\coordinate (yz) at (0.38,-0.12);
\coordinate (zz) at (-0.58,-0.80);
\coordinate (iz) at (0.47,0.65);
\coordinate (zi) at (0.23,0.32);
\coordinate (yx) at (0.58,-0.80);
\coordinate (ix) at (-0.47,0.65);
\coordinate (yi) at (-0.23,0.32);
\coordinate (xz) at (0.95,0.30);
\coordinate (yy) at (-0.76,-0.24);
\coordinate (zx) at (-0.38,-0.12);
\draw (xi) -- (ix) -- (xx);
\draw [style=double] (xx) -- (yy) -- (zz);
\draw [style=double] (yx) -- (zy) -- (xz);
\draw (xz) -- (iz) -- (xi);
\draw (zz) -- (zi) -- (iz);
\draw (ix) -- (yi) -- (yx);
\draw (xz) -- (zx) -- (yy);
\draw (xx) -- (yz) -- (zy);
\draw (xi) -- (iy) -- (xy);
\draw (zz) -- (xy) -- (yx);
\draw plot coordinates{(zx) (ix) (zi)};
\draw plot coordinates{(yi) (iz) (yz)};
\draw plot coordinates{(zi) (zy) (iy)};
\draw [style=double] plot coordinates{(yz) (xy) (zx)};
\draw plot coordinates{(iy) (yy) (yi)};

\node[draw,circle,dashed] at (xi) {XI};
\node[draw,dashed] at (xy) {XY};
\node at (iy) {IY$_3$};
\node at (xx) {XX$_1$};
\node[draw] at (zy) {ZY};
\node at (yz) {YZ$_4$};
\node at (zz) {ZZ$_7$};
\node[draw,dashed] at (iz) {IZ};
\node at (zi) {ZI$_6$};
\node at (yx) {YX$_5$};
\node[draw,dashed] at (ix) {IX};
\node at (yi) {YI$_9$};
\node at (xz) {XZ$_2$};
\node[draw] at (yy) {YY};
\node at (zx) {ZX$_8$};
\end{tikzpicture} \end{center}
\caption{The $2$-qubit doily $W_2$, the ovoid $O_2$ (framed), the triad $T_2$ 
  (framed and dashed), its center $c_2$ (circled and dashed) and the completion
  order (subscripted). The negative lines are doubled.}
\label{co-fig}
\end{figure}

The algorithm relies on the following predefined elements, depicted in
Fig.~\ref{co-fig}:
   the 2-qubit doily $W_2$,
  the ovoid $O_2 \equiv \{IX,IZ,XY,ZY,YY\}$ in $W_2$,
  the unicentric triad $T_2 \equiv \{IX,IZ,XY\}$ in $O_2$,
  the center $c_2 \equiv XI$ of $T_2$, and
  the sequence of lines
\begin{align*}
S \equiv & \ (XI,IX,XX), (XI,IZ,XZ), (XI,XY,IY), (ZY,XX,YZ), (ZY,XZ,YX), (ZY,IY,ZI),
\\
 & \ (YY,XX,ZZ), (YY,XZ,ZX), (YY,IY,YI).
\end{align*} In the figure the third element of the tuples in this
\emph{completion order} is numbered from 1 to 9.

\begin{algorithm}[!ht]
  \begin{center}
    \begin{algorithmic}[1]
      \ForEach{ovoid $O = \{o_1, o_2, o_3, o_4, o_5\}$ in $W_N$, with $o_1 < o_2 < o_3 < o_4 < o_5$}
          \State $f(IX) \gets o_1 \mid\mid  f(IZ) \gets o_2 \mid\mid  f(XY) \gets o_3 \mid\mid  f(ZY) \gets o_4 \mid\mid  f(YY) \gets o_5$\label{lst:line:ovoid}
          \ForEach{center $c$ of $\{o_1, o_2, o_3\}$ in $W_N$ that anticommutes with $o_4$ and $o_5$} \label{lst:line:search:center}
            \State $f(c_2) \gets c$  \label{lst:line:center}
            \IFor{line $(p,q,r)$ in the order of the sequence $S$}{$f(r) \gets |f(p).f(q)|$} \label{lst:line:completion}
            \IIf{$O$ is not the smallest ovoid of $f$} discard $f$ \EndIIf \label{lst:line:discard}
            \State \ldots \hspace{2cm} $\triangleright$ location for a potential treatment of $f$ \label{lst:line:treatment}
          \EndFor
      \EndFor
    \end{algorithmic}
  \end{center}
  \caption{Doily generation algorithm. \label{algo-fig}}
\end{algorithm}

The algorithm itself is presented in Algorithm~\ref{algo-fig}, where $f(a) \gets b$
denotes the assignment of $b$ as the image of $a$ by $f$.

On Line~\ref{lst:line:ovoid} a doily labeling $f$ is partially defined by the
choice of images for the 5 points of the ovoid $O_2$ of $W_2$. These images are
the points of some ovoid $O = \{o_1, o_2, o_3, o_4, o_5\}$ of $N$-qubit
observables.  The points are assigned in increasing order so that to avoid
duplicates.
As these five assignments are independent, they can be performed in parallel.

Then (on Line~\ref{lst:line:search:center}) the algorithm looks for a point $c$
that commutes with the first three points of $O$
 and that anticommutes with its last two points $o_4$ and $o_5$. On
 Line~\ref{lst:line:center} this point becomes the image by $f$ of
the center $c_2$ of the triad $T_2$ of $O_2$.

The completion step on Line~\ref{lst:line:completion} computes one by one the
images of all the other points of $W_2$ by $f$, in the order described by the
sequence of lines $S$. At each iteration of this loop, for the line $(p,q,r)$,
the values $f(p)$ and $f(q)$ are known. By definition of a doily line, the image
by $f$ of the third point $r$ is the product of the images $f(p)$ and $f(q)$ of
the first two points, up to a possible minus sign, removed by the operation
$|\_|$ that denotes absolute value.

Knowing that each doily features 6 ovoids, the same doily is generated 6 times
before Line~\ref{lst:line:discard}, whose statement keeps only one of them,
namely the doily $d$ generated from the ovoid that is the smallest (according to
the lexicographic order) among the 6 ovoids in $d$.

On Line~\ref{lst:line:treatment} various treatments of the generated doilies $f$
can be added, such as a storage, or the computation of classification criteria
 defined in Section~\ref{classify-sec}.

\subsection{Justification of the generation algorithm}
\label{justif-sec}

First of all, the fact that doily labelings encode multi-qubit doilies is a
direct consequence of the definition of a multi-qubit doily. Then, the
properties of correctness and completeness for the doily enumeration algorithm
mainly come from the following definition and proposition, whose proof is
illustrated by Figure~\ref{fig:doily-completion-lines}.

\begin{figure}[ht!]
\begin{center}
\begin{tikzpicture}[every plot/.style={smooth, tension=2},
  scale=4.4,
]
\coordinate (xi) at (0,1.0);
\coordinate (xy) at (0,-0.80);
\coordinate (iy) at (0,-0.40);
\coordinate (xx) at (-0.95,0.30);
\coordinate (zy) at (0.76,-0.24);
\coordinate (yz) at (0.38,-0.12);
\coordinate (zz) at (-0.58,-0.80);
\coordinate (iz) at (0.47,0.65);
\coordinate (zi) at (0.23,0.32);
\coordinate (yx) at (0.58,-0.80);
\coordinate (ix) at (-0.47,0.65);
\coordinate (yi) at (-0.23,0.32);
\coordinate (xz) at (0.95,0.30);
\coordinate (yy) at (-0.76,-0.24);
\coordinate (zx) at (-0.38,-0.12);
\draw [ultra thick] (xi) -- (ix) -- (xx);
\draw [ultra thick] (xx) -- (yy) -- (zz);
\draw [ultra thick] (yx) -- (zy) -- (xz);
\draw [ultra thick] (xz) -- (iz) -- (xi);
\draw (zz) -- (zi) -- (iz);
\draw (ix) -- (yi) -- (yx);
\draw [ultra thick] (xz) -- (zx) -- (yy);
\draw [ultra thick] (xx) -- (yz) -- (zy);
\draw [ultra thick] (xi) -- (iy) -- (xy);
\draw (zz) -- (xy) -- (yx);
\draw plot coordinates{(zx) (ix) (zi)};
\draw plot coordinates{(yi) (iz) (yz)};
\draw [ultra thick] plot coordinates{(zi) (zy) (iy)};
\draw plot coordinates{(yz) (xy) (zx)};
\draw [ultra thick] plot coordinates{(iy) (yy) (yi)};

\node[draw,circle,dashed] at (xi) {$c$};
\node[draw,dashed] at (xy) {$o_3$};
\node at (iy) {$\pm c.o_3$};
\node at (xx) {$\pm c.o_1$};
\node[draw] at (zy) {$o_4$};
\node at (yz) {$\pm (c.o_1).o_4$};
\node at (zz) {$\pm (c.o_1).o_5$};
\node[draw,dashed] at (iz) {$o_2$};
\node at (zi) {$\pm (c.o_3).o_4$};
\node at (yx) {$\pm (c.o_2).o_4$};
\node[draw,dashed] at (ix) {$o_1$};
\node at (yi) {$\pm (c.o_3).o_5$};
\node at (xz) {$\pm c.o_2$};
\node[draw] at (yy) {$o_5$};
\node at (zx) {$\pm (c.o_2).o_5$};
\end{tikzpicture} \end{center}
\caption{Expression of each observable according to the completion order followed by the generation 
algorithm, from the doily root $(\{o_1, o_2, o_3, o_4, o_5\},c)$. The thick lines are the lines used
 to compute these expressions. \label{fig:doily-completion-lines}}
\end{figure}
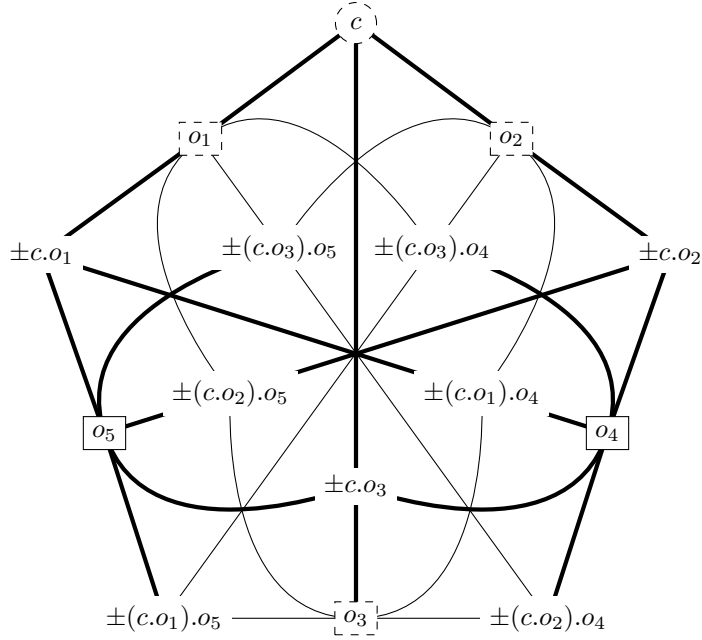

\begin{definition}[Doily Root]
\label{root-def}
Any pair $(O,c)$ such that $O$ is an ovoid of $W_N$ and $c$ is a point of $W_N$ that
commutes with exactly three points of $O$ and anticommutes with the other two
points is called a \emph{(multi-qubit) doily root}.
\end{definition}

\begin{proposition}
\label{doily6deter}
Any doily root $(O,c)$ of $W_N$ determines exactly one $N$-qubit doily.
\end{proposition}

\begin{proof}
Let $O=\{o_1,o_2,o_3,o_4,o_5\}$ and $c$ be such that $(O,c)$ is a multi-qubit
doily root. The fact that $c$ commutes with $o_1$, $o_2$ and $o_3$ implies that
there exist three multiplicative factors $a_1$, $a_2$ and $a_3$ $\in \{-1,1\}$
such that $\{c,o_i,a_i~c.o_i\}$ are isotropic lines of $W_N$, for $i=1,2,3$.
This is depicted in Figure~\ref{fig:doily-completion-lines} by the 3 points $\pm
c.o_i$. But because $c$ anticommutes with $o_4$ and $o_5$ we also have that the
observables $a_i~c.o_i$ commute with $o_4$ and $o_5$. Therefore there exist
multiplicative factors $a_{ij}\in \{-1,1\}$ such that
$\{a_i~c.o_i,o_j,a_{ij}~(c.o_i).o_j\}$ are isotropic lines, for $i=1,2,3$ and
$j=4,5$. This is depicted in Figure~\ref{fig:doily-completion-lines} by the 6
points of the form $\pm (c.o_i).o_j$. These 9 points are computed by the
completion step of the generation algorithm, in the order $a_1~c.o_1$,
$a_2~c.o_2$, $a_3~c.o_3$, $a_{14}~(c.o_1).o_4$, $a_{24}~(c.o_2).o_4$,
$a_{34}~(c.o_3).o_4$, $a_{15}~(c.o_1).o_5$, $a_{25}~(c.o_2).o_5$,
$a_{35}~(c.o_3).o_5$. The 9 geometric lines thus identified are depicted by
thick lines in Figure~\ref{fig:doily-completion-lines}. Finally, it is easy to
check that the six 3-sets represented by the thin lines in
Figure~\ref{fig:doily-completion-lines} indeed are geometric lines.
A noticeable property is that the product of the three points on each of these
lines contains twice the center $c$ and once each point of the ovoid $O$. By
applying the known commutation and anticommutation relations between these
points, it comes that the product of both centers annihilates. So, modulo a
possible minus sign, it remains the product of all observables of the ovoid,
known to equal identity. Therefore, the product of the three observables on each
line equals $\pm{\cal I}_N$. Consequently, these 15 points and 15 lines form a
doily, shown in Figure~\ref{fig:doily-completion-lines}, so the algorithm is
correct.

Each multi-qubit doily features at least one ovoid and the first loop explores
all ovoids in $W_N$. So, each doily is found six times before Line~\ref{lst:line:discard},
 since each multi-qubit doily features exactly six ovoids. As the statement
on this line always keeps one of them (the one that has been produced from the
smallest of its ovoids), the algorithm is also complete.\end{proof}

\subsection{Algorithmic complexity and implementation details}

The enumeration algorithm explores all 4-tuples of observables likely to form an
ovoid (the fifth point in the ovoid is computed as the product of the previous
four), and then explores all observables to find $c$ (on
Line~\ref{lst:line:search:center} of Algorithm~\ref{algo-fig}).
Therefore, the complexity of the algorithm is estimated to be
$O\left(4^{5N}\right)$, when the time unit is the duration to check whether two
observables commute.

For efficiency reasons, we have implemented the algorithm in the C language,
which allows for many optimizations.
The total code is composed of about \numprint{2300} lines and 50 functions, some
of which implementing the classification process presented in
Section~\ref{classify-sec}. Some factors make the algorithm implementation more
efficient than the former one presented in~\cite{SdHG21}: The new algorithm has
a lower complexity; compared to the previously used language
Magma~\cite{BCP97}, the low-level language C allows to perform fast
operations on bit vector representations of the observables, using as few CPU
instructions as necessary, and to split the workload into multiple
threads.

The calculations were run on Linux Ubuntu, on a PC equipped with an Intel
(R) Core(TM) i7-8665U 1.90~GHz and 15~GB RAM. The code was compiled with gcc
9.3.0 with optimization Ofast and is multi-threaded with OpenMP.

\section{Multi-doily classification process and results}
\label{classify-sec}

This section presents our classification criteria of $N$-qubit doilies
 and the classification results for $N = 4$ and $N = 5$.

\subsection{Classification criteria}

The classification parameters adopted are the same as in~\cite{SdHG21}. The
classification of an $N$-qubit doily is based on: 1) its signature, i.\,e. the
number of its observables containing a given number of $I$ : $N-1$, $N-2$,
$N-3$, \ldots respectively named types $A$, $B$, $C$, \ldots; 2) the
configuration of its negative lines, as described in Section~\ref{basic-sec};
and 3) its linear or quadratic character.

To find the line configuration of a doily, the first discriminatory factor is
the number of negative lines, since for each number of negative lines except 7
and 8, there is only one configuration possible. Then the property used to
distinguish configurations 7A from 7B and 8A from 8B is to count the number of
observables contained in at least one negative line, since this number is
different between A and B.

We use the following property to check whether a doily is linear or quadratic.
Given an $N$-qubit doily, we pick up in it a tricentric triad (here we take the
image of $\{XY,ZY,YI\}$). If the product of the corresponding three observables
is $\pm i {\cal I}_N$, then the doily is linear, otherwise it is quadratic. This is
because any tricentric triad is a line in the ambient PG$(3,2)$ if a doily spans
a PG$(3,2)$.

\subsection{Database of numerical results}

Using the program described in Section~\ref{sec:enum-sec}, we were able to classify
all doilies for $N=3$ (\numprint{2016} ovoids), $N=4$ (\numprint{548352} ovoids)
and $N=5$ (\numprint{142467072} ovoids). This classification is a treatment
added on Line~\ref{lst:line:treatment} of the algorithm presented in
Algorithm~\ref{algo-fig}, that determines the complete type of each generated
doily, counts the number of doilies for each type, and registers it in a result
table.

The sums of the numbers of linear and quadratic doilies found in each of the
above-mentioned cases correspond exactly to those stemming from eqs.~(\ref{lN})
and~(\ref{qN}), respectively, summarized in Table~\ref{nbdoilies-tab}.
The results of our classification are collected in Appendix~\ref{res3} (three
qubits), Appendix~\ref{res4} (four qubits) and Appendix~\ref{res5} (five
qubits). The data for three qubits are in complete agreement with those
of~\cite{SdHG21}; we found 11 different types of doilies of which five are
linear and six quadratic. The 95 distinct types of four-qubit doilies split into
24 linear and 71 quadratic ones, whereas amongst 447 types of five-qubit doilies
one finds 89 linear and 358 quadratic.

The structure of the classification table in each appendix is the same: the
first column gives the type, the next $N$ columns feature the numbers of
observables of the corresponding types in a doily of the given type, the $\nu$
column shows the doily's character, and the remaining columns contain
information about how many doilies of the given type are endowed with a
particular number of negative lines (the blank space stands for zero here).
The types are ordered in decreasing order of the number of observables
containing no $I$s, in case of equality in decreasing order of the number of
observables containing one $I$, and so on up to the number $A$ of observables
containing $N-1$ $I$s. For instance, for 4 qubits, the type 1 contains the
maximal number 12 of $D$-type observables, and the last type 95 contains only
$A$- and $B$-type observables. For a given signature, the type of quadratic
doilies precedes that of linear ones.

The result tables are stored in \url{https://quantcert.github.io/}.

The C code for classification runs in 0.3~s
 for 4 qubits with 1.4 MB of memory and 12~min with 1.8 MB of memory for 5 qubits.
 The memory usage is low because the doilies are 
not stored, all the measurements are performed on the fly.

\subsection{Remarks about five-qubit doilies}

Let us have a closer look at the five-qubit case. The $3^2\times\binom{5}{2}=90$
observables of type $B$ and $3^4\times\binom{5}{4}=405$ observables of type $D$
lie on an elliptic quadric $\mathcal{Q}^-_{(YYYYY)}(9,2)$ of $\mathcal{W}(9,2)$. 
This special quadric $\mathcal{Q}^-_{(YYYYY)}(9,2)$, like any non-degenerate
quadric, is a \emph{geometric hyperplane} of $\mathcal{W}(9,2)$. As a doily is
also a \emph{subgeometry} of $\mathcal{W}(9,2)$, it either lies fully in $
\mathcal{Q}^{-}_{(YYYYY)}(9,2)$ (in which case $B \cup D = 15$, such a doily will
be called special), or shares with $\mathcal{Q}^{-}_{(YYYYY)}(9,2)$ a set of
points that form a geometric hyperplane, in particular an ovoid ($B\cup D = 5$),
a perp-set ($B \cup D = 7$) and/or a grid ($B \cup D = 9$) and being referred to
as ovoidal, perpial and/or gridal, respectively.

From Appendix~\ref{res5} one can infer a number of interesting properties. We
first notice that signatures with $B \cup C$ being even or odd are endowed with
even or odd numbers of negative lines, respectively.

We also observe that there are 12 different signatures with $A=C=E=0$, i.\,e.,
signatures featuring solely special doilies. 

Further, there are 17 particular signatures such that each features observables
of every type and no two types have the same cardinality.
Out of them, six are ovoidal (e.\,g., 2-1-3-4-5), seven perpial (e.\,g.,
1-3-2-4-5) and four gridal (e.\,g., 2-4-3-5-1).

If all doilies of a particular signature have just five or just six negative
lines, then each doily is ovoidal; if a signature features just seven negative
lines, then all of its doilies are perpial.

Among 33 distinct signatures with four negative lines only, one finds 12 ovoidal,
11 perpial and 9 gridal ones; doilies of the remaining signature, viz.
0-8-0-7-0, are special.

Next, there are 15 different signatures whose doilies are endowed  with 12
(i.\,e., the maximum number of) negative lines. Out of them, five are ovoidal,
five perpial and four gridal; the doilies of the remaining signature, namely
0-0-0-15-0, are special. Similarly, there are 35 distinct signatures whose
doilies contain 11 (i.\,e., the maximum odd number of) negative lines; out of
them, 10 are ovoidal, 11 perpial and 12 gridal, with the remaining two
signatures, viz. 0-1-0-14-0 and 0-3-0-12-0, featuring solely special doilies.

\subsection{Specific behavior of linear doilies}

Finally, this section and the next one briefly mention some properties of linear
doilies. Like the three- and four-qubit cases, a linear five-qubit doily can be
either ovoidal or gridal and always contains an odd number of negative lines.
Also, 75 types of linear doilies share their signatures with their quadratic siblings.
However, there are 14 different signatures that are genuinely linear, of which
eight cases are ovoidal.

From our results on three-, four- and five-qubit cases it follows that a linear
doily (a) always features an odd number of negative lines, and (b) does not
share a perp-set with the distinguished quadric.

We conjecture that Property (a) holds for any number of qubits $N \geq 2$, but we
have not yet found of proof of it; we surmise that it has something to do with
the fact that a linear doily is ``squeezed'' into a PG$(3,2)$, compared to a
quadratic doily that enjoys more degrees of freedom being stretched out in a 
PG$(4,2)$. Property (b) can readily be proved to hold for any  $N \geq 3$, as follows. 

\begin{proof} 
Let us consider a linear doily with one of its perp-set; on 
Figure~\ref{doily-uni-triad}, this perp-set is illustrated in bold font. Any
perp-set of any doily features four tricentric triads; in our perp-set these
triads are: $\{1,2,3\}$, $\{3,4,5\}$, $\{1,5,6\}$ and $\{2,6,4\}$. Now, we know
that any tricentric triad of a linear doily corresponds to a non-isotropic line
in the ambient projective space, the four lines plus the three (totally
isotropic) lines of the perp-set forming a Fano plane in this space, which is
illustrated in Figure~\ref{fano-plane}. The assumption that our perp-set also
lies on the distinguished quadric would mean that the whole plane would lie in
the distinguished quadric and so would be totally isotropic, a contradiction.
\end{proof}

\begin{figure}[hbt!]
\begin{subfigure}{0.5\textwidth}
\begin{center}
\begin{tikzpicture}[every plot/.style={smooth, tension=2},
  scale=1.7,
  every node/.style={scale=0.7,fill=white,circle,draw=black}
]
\coordinate (25) at (0,1.0);
\coordinate (34) at (0,-0.80);
\coordinate (16) at (0,-0.40);
\coordinate (36) at (-0.95,0.30);
\coordinate (24) at (0.76,-0.24);
\coordinate (15) at (0.38,-0.12);
\coordinate (12) at (-0.58,-0.80);
\coordinate (46) at (0.47,0.65);
\coordinate (35) at (0.23,0.32);
\coordinate (56) at (0.58,-0.80);
\coordinate (14) at (-0.47,0.65);
\coordinate (23) at (-0.23,0.32);
\coordinate (13) at (0.95,0.30);
\coordinate (45) at (-0.76,-0.24);
\coordinate (26) at (-0.38,-0.12);
\draw [ultra thick] (25) -- (14) -- (36);
\draw (36) -- (45) -- (12);
\draw (12) -- (34) -- (56);
\draw (56) -- (24) -- (13);
\draw [ultra thick] (13) -- (46) -- (25);
\draw (12) -- (35) -- (46);
\draw (14) -- (23) -- (56);
\draw (13) -- (26) -- (45);
\draw (36) -- (15) -- (24);
\draw [ultra thick] (25) -- (16) -- (34);
\draw plot coordinates{(26) (14) (35)};
\draw plot coordinates{(23) (46) (15)};
\draw plot coordinates{(35) (24) (16)};
\draw plot coordinates{(15) (34) (26)};
\draw plot coordinates{(16) (45) (23)};
\node[double] at (25) {$\mathbf{0}$};
\node at (34) {$\mathbf{6}$};
\node at (16) {$\mathbf{3}$};
\node at (36) {$\mathbf{4}$};
\node at (24) {$\;\;$};
\node at (15) {$\;\;$};
\node at (12) {$\;\;$};
\node at (46) {$\mathbf{2}$};
\node at (35) {$\;\;$};
\node at (56) {$\;\;$};
\node at (14) {$\mathbf{1}$};
\node at (23) {$\;\;$};
\node at (13) {$\mathbf{5}$};
\node at (45) {$\;\;$};
\node at (26) {$\;\;$};
\end{tikzpicture} \end{center}
\caption{}
\label{doily-uni-triad}
\end{subfigure}
\begin{subfigure}{0.5\textwidth} 
\begin{center}
\begin{tikzpicture}[every plot/.style={smooth, tension=2},
  scale=1.2,
  every node/.style={scale=0.7,fill=white,circle,draw=black}
]
\coordinate (0) at (0,-0.38);
\coordinate (1) at (-1,-1);
\coordinate (2) at (1,-1);
\coordinate (3) at (0,-1);
\coordinate (4) at (0.5,0);
\coordinate (5) at (-0.5,0);
\coordinate (6) at (0,1);
\draw [ultra thick] (6) -- (0) -- (3);
\draw (1) -- (5) -- (6);
\draw (6) -- (4) -- (2);
\draw (1) -- (3) -- (2);
\draw [ultra thick] (1) -- (0) -- (4);
\draw [ultra thick] (5) -- (0) -- (2);
\draw (0) circle[radius=0.62];
\node[double] at (0) {$0$};
\node at (1) {$1$};
\node at (2) {$2$};
\node at (3) {$3$};
\node at (4) {$4$};
\node at (5) {$5$};
\node at (6) {$6$};
\end{tikzpicture} \end{center}
\caption{}
\label{fano-plane}
\end{subfigure}
\caption{Graphical arguments for the property that a linear
doily cannot
share a perp-set with the distinguished quadric.}
\label{fig:proof-linearity}
\end{figure}
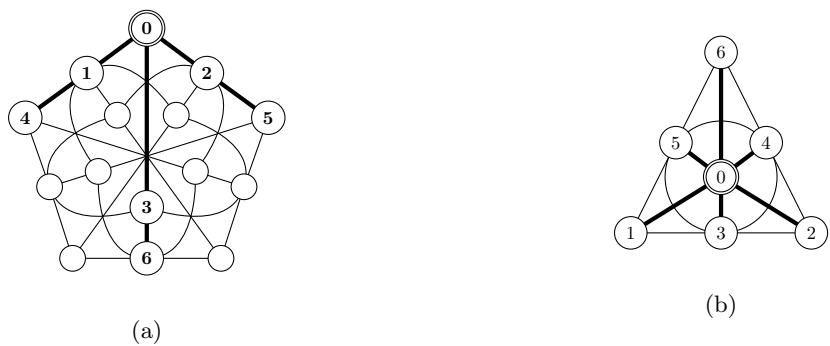

\subsection{A distinguished hexad of (linear) doilies}

One knows that given an ovoid, there is a unique linear doily containing this
ovoid. Now, take any quadratic doily. As each of its six ovoids defines a unique
linear doily, we have a unique hexad of doilies tied to each quadratic doily.
This holds for any $N \geq 3$. Figure~\ref{fig:hexad-doilies} illustrates this
property for $N=3$. It features a quadratic doily in the middle, its six ovoids
depicted explicitly as pentads of points located on bold gray lines and the
corresponding six linear doilies; for better readability, the
 points of the corresponding ovoids are illustrated by double-circles.

\begin{figure}[!ht]
  \begin{center}
    \includegraphics[scale=0.9]{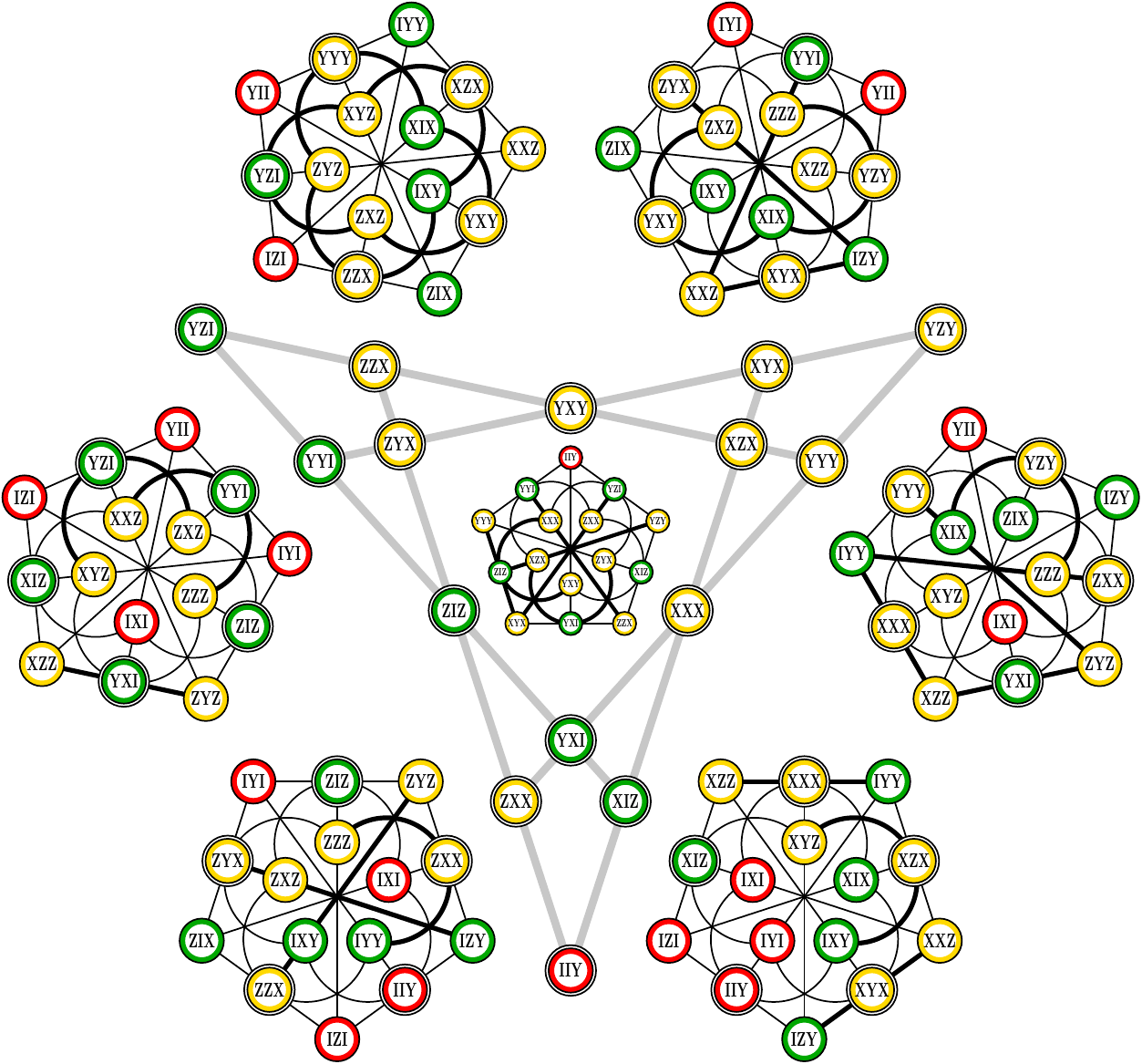}
    \caption{A particular hexad of linear doilies in the three-qubit symplectic 
      polar space.}
    \label{fig:hexad-doilies}
  \end{center}
\end{figure}

\section{Conclusion}
\label{conclusion-sec}

There are a number of intriguing extensions and generalizations of the ideas and
findings presented in this paper. We shall mention a few of them.

An interesting situation that will be worth addressing occurs in the case of $N=4$.
Given a PG$(3,2)$ of the ambient PG$(7,2)$ of $\mathcal{W}(7,2)$, its polar space 
is another PG$(3,2)$. Hence, PG$(3,2)$s in PG$(7,2)$ come in polar pairs. Taking 
into account the fact that a non-isotropic PG$(3,2)$ features a unique linear doily 
of $\mathcal{W}(7,2)$, the above property means that also linear doilies 
of $\mathcal{W}(7,2)$ occur in pairs. That is, picking up any linear four-qubit doily,
there exists a unique linear doily such that each of its 15 observables commutes
with each observable of the selected doily. 
This observation raises several interesting questions. For example, it would be interesting
to ascertain which signatures are/can be paired, or which cardinalities of negative
lines can occur in such self-polar pairs; we have already checked by hand a few examples 
where both doilies in a pair have the same signature and feature the same number of negative lines.
There are (see Appendix B) altogether 24 different signatures featured
by linear four-qubit doilies. We can then create a graph on 24 vertices such that
its two edges are connected if there exists a pair of linear doilies
exhibiting the corresponding signatures; we can even add a weight to
an edge showing how many pairs of doilies feature this particular pair
of signatures. This graph, as it follows from the examples checked, will also
have edges joining a vertex to itself when the two paired signatures are identical.
So, being an interesting graph of its own, it will also reveal
some finer traits of the relation between individual linear doilies in $\mathcal{W}(7,2)$!

A particular case deserving closer attention is $N=6$. Here, let us formally view any six-qubit
observable as a `syntheme' partitioned into three two-qubit observables (`duads'). Given a partition,
we find a set of linear doilies such that any doily in the set features 15 particular observables 
such that when restricted to the same duad we get a two-qubit doily; that is, any such doily can formally
be regarded as being composed of three two-qubit doilies. Moreover, each partition features a prominent
doily having all the three duads identical. The next worth-exploring case in this respect is $N=9$, 
as $\mathcal{W}(17,2)$ hosts
not only composites comprising three doilies having the same number of qubits (namely three), but also those
whose compounds feature different numbers of qubits (namely four, three and two).

Another prospective, but much more challenging, task will be to count and classify all rank-three spaces, 
$\mathcal{W}(5,2)$s, living in a particular $W_N$, for $N \geq 4$. The case $N = 4$ was already briefly examined
in \cite{SdHG21}. To address higher rank cases, we plan to employ the strategy that is the direct and natural generalization of the
ovoid-based algorithm for doilies described in this paper. Geometrically, an $N$-qubit ovoid is a set of five points lying on a certain elliptic 
quadric of a PG$(3,2)$ in the ambient PG$(2N-1,2)$. Hence, its analogue will be a set of 27 $N$-qubit observables lying
on an elliptic quadric of PG$(5,2)$ in the ambient PG$(2N-1,2)$, and a triad of the ovoid will have its counterpart in
a quadratic doily located on the quadric. A root of an $N$-qubit $\mathcal{W}(5,2)$ will thus comprise an elliptic quadric and an off-quadric point
such that its associated observable commutes with each of the 15 observables of a doily located in the quadric and anticommutes
with the remaining 12 observables. It is obvious that this task will require a more elaborate generation algorithm and a more complex computer code
to be successfully accomplished.

\section*{Acknowledgments}
\label{sec:acknowledgments}

This project is supported by the EIPHI Graduate School (contract
ANR-17-EURE-0002). This work was also supported by the Slovak VEGA Grant Agency,
Project \# 2/0004/20. We thank our friend Zsolt Szab\'o for the help in 
preparation of several figures. 

\bibliographystyle{unsrt}

\newpage

\appendix

\section{Taxonomy of 3-qubit doilies}
\label{res3}

\vspace*{0.5cm}
{
\begin{longtable}{|c|c|c|c|c|c|c|c|c|c|c|c|c|c|c|c|c|} \cline{2-4} \cline{6-17}
 \multicolumn{1}{c|}{} & \multicolumn{3}{c|}{Observables} &  & \multicolumn{12}{c|}{Configuration of negative lines}  \\
\cline{1-17}

Type & $A$  & $B$  & $C$  & $\nu$ & 3 & 4 & 5 & 6 & 7A & 7B & 8A & 8B & 9 & 10 & 11 & 12 \\
\endhead
\hline 1 & 1 & 5 & 9  & $q$ &  &  &  & 108 &  &  &  &  &  &  &  &  \\
 \hline 2 & 0 & 7 & 8  & $q$ &  &  &  &  & 81 &  &  &  &  &  &  &  \\
 \hline 3 & 2 & 5 & 8  & $l$ &  &  & 162 &  &  &  &  &  &  &  &  &  \\
 \hline 4 & 3 & 5 & 7  & $q$ &  & 324 &  &  &  &  &  &  &  &  &  &  \\
 \hline 5 & 0 & 9 & 6  & $l$ & 9 &  &  &  & 27 &  &  &  &  &  &  &  \\
 \hline 6 & 2 & 7 & 6  & $q$ & 216 &  & 162 &  &  &  &  &  &  &  &  &  \\
 \hline 7 & 4 & 5 & 6  & $l$ & 54 &  &  &  &  &  &  &  &  &  &  &  \\
 \hline 8 & 2 & 9 & 4  & $l$ & 81 &  &  &  &  &  &  &  &  &  &  &  \\
 \hline 9 & 4 & 7 & 4  & $q$ & 81 &  &  &  &  &  &  &  &  &  &  &  \\
 \hline 10 & 0 & 15 & 0  & $q$ & 36 &  &  &  &  &  &  &  &  &  &  &  \\
 \hline 11 & 6 & 9 & 0  & $l$ & 3 &  &  &  &  &  &  &  &  &  &  &  \\
 \hline \hline
\end{longtable}}
 
\section{Taxonomy of 4-qubit doilies}
\label{res4}

{\setlength{\tabcolsep}{0.4em}{\begin{longtable}{|c|c|c|c|c|c|c|c|c|c|c|c|c|c|c|c|c|c|} \cline{2-5} \cline{7-18}
    \multicolumn{1}{c|}{} & \multicolumn{4}{c|}{Observables} &  & \multicolumn{12}{c|}{Configuration of negative lines}  \\
   \cline{1-18}
   
   Type & $A$  & $B$  & $C$  & $D$  & $\nu$ & 3 & 4 & 5 & 6 & 7A & 7B & 8A & 8B & 9 & 10 & 11 & 12 \\
   \endhead
 \hline 1 & 0 & 3 & 0 & 12  & $q$ & 216 &  &  &  & 648 &  &  &  & 648 &  &  &  \\%
 \hline 2 & 0 & 4 & 0 & 11  & $q$ &  &  &  & 3888 &  &  & 3888 &  &  &  &  &  \\%
 \hline 3 & 0 & 5 & 0 & 10  & $q$ & 972 &  & 1944 &  & 4860 & 1944 &  &  & 1944 &  &  &  \\%
 \hline 4 & 1 & 0 & 5 & 9  & $q$ & 648 &  &  &  &  &  &  &  & 648 &  &  &  \\%
 \hline 5 & 3 & 0 & 3 & 9  & $l$ & 144 &  &  &  &  &  &  &  &  &  &  &  \\%
 \hline 6 & 0 & 6 & 0 & 9  & $q$ &  & 1296 &  & 5184 &  &  &  &  &  &  &  &  \\%
 \hline 7 & 0 & 1 & 6 & 8  & $q$ & 972 &  &  &  & 3888 &  &  &  &  &  & 972 &  \\%
 \hline 8 & 1 & 1 & 5 & 8  & $q$ &  &  &  & 7776 &  &  &  &  &  &  &  &  \\
\hline \hline
   \end{longtable}}}
 
\newpage

\section{Taxonomy of 5-qubit doilies}
\label{res5}

{
  \footnotesize
  \setlength{\tabcolsep}{0.02em} %

  {\begin{longtable}{|c|c|c|c|c|c|c|c|c|c|c|c|c|c|c|c|c|c|c|} \cline{2-6} \cline{8-19}
    \multicolumn{1}{c|}{} & \multicolumn{5}{c|}{Observables} &  & \multicolumn{12}{c|}{Configuration of negative lines}  \\
   \cline{1-19}
   
  Type  & $A$  & $B$  & $C$  & $D$  & $E$  & $\nu$ & 3 & 4 & 5 & 6 & 7A & 7B & 8A & 8B & 9 & 10 & 11 & 12 \\
   \endhead
   \hline 1 & 0 & 0 & 1 & 5 & 9  & $q$ &  &  &  &  &  & 58\,320 &  &  & 19\,440 &  &  &  \\%
   \hline 2 & 1 & 0 & 0 & 5 & 9  & $q$ &  &  &  & 12\,960 &  &  &  &  &  &  &  &  \\%
   \hline 3 & 0 & 1 & 1 & 4 & 9  & $q$ &  & 58\,320 &  & 233\,280 &  &  &  & 233\,280 &  & 116\,640 &  & 19\,440 \\%
   \hline 4 & 0 & 2 & 1 & 3 & 9  & $q$ & 68\,040 &  & 174\,960 &  & 116\,640 & 466\,560 &  &  & 262\,440 &  & 116\,640 &  \\%
   \hline 5 & 1 & 2 & 0 & 3 & 9  & $q$ &  &  &  & 12\,960 &  &  &  &  &  &  &  &  \\%
   \hline 6 & 0 & 3 & 1 & 2 & 9  & $q$ &  & 116\,640 &  & 421\,200 &  &  & 116\,640 & 174\,960 &  & 116\,640 &  &  \\%
   \hline 7 & 0 & 4 & 1 & 1 & 9  & $q$ &  &  & 29\,160 &  & 58\,320 & 29\,160 &  &  &  &  &  &  \\%
   \hline 8 & 0 & 5 & 1 & 0 & 9  & $q$ &  &  &  & 9720 &  &  &  &  &  &  &  &  \\
\hline \hline
   \end{longtable}
}}

\end{document}